\documentclass{article}
\usepackage[english]{babel}
\usepackage[english]{babel}
\usepackage{lineno} 
\usepackage{graphicx,multicol}
\usepackage{epic,eepic,epsfig}
\usepackage{amssymb}
\usepackage{amsmath,amsthm}
\usepackage{color}
\usepackage{tikz}
\newcommand{\jbj}[1]{{\color{red}#1}}

\usepackage{eso-pic} % Magic that gets rid the newline errors in the pics

\newcommand{\dXY}[2]{d(#1,#2)}  % The number of arcs from #1 to #2 

\newcommand{\ssX}{\hspace{1.2cm}}

\newcommand{\induce}[2]{\mbox{$ #1 \langle #2 \rangle$}}

\newcommand{\dom}{\mbox{$\rightarrow$}}

\newcommand{\sdom}{\mbox{$\mapsto$}}

\newcommand{\2}{\vspace{2mm}}

\theoremstyle{plain}
\newtheorem{theorem}{Theorem}
\newtheorem{proposition}[theorem]{Proposition}
\newtheorem{claim}{Claim}[theorem]

\newtheorem{corollary}[theorem]{Corollary}
\newtheorem{lemma}[theorem]{Lemma}

\theoremstyle{definition}
\newtheorem{remark}[theorem]{Remark}

\newtheorem{conjecture}[theorem]{Conjecture}

\newenvironment{subproof}{\par\noindent {\it Proof}.\ }{\hfill$\lozenge$\par\vspace{11pt}}

\DeclareMathOperator{\ind}{ind}

\begin{document}
\bibliographystyle{plain}
%\pagewiselinenumbering
%\setpagewiselinenumbers
%\modulolinenumbers[1]
%\linenumbers

\title{Spanning eulerian subdigraphs in semicomplete digraphs\thanks{Research supported by the Danish research council under grant number DFF 7014-00037B and DISCO project, PICS, CNRS.}}
\author{J. Bang-Jensen\thanks{Department of Mathematics and Computer Science, University of Southern Denmark, Odense DK-5230, Denmark (email:jbj@imada.sdu.dk). } \and Fr\'ed\'eric Havet\thanks{CNRS, Universit\'e C\^ote d'Azur, I3S and INRIA, Sophia Antipolis, France (email: frederic.havet@inria.fr)
} \and Anders Yeo\thanks{Department of Mathematics and Computer Science, University of Southern Denmark, Odense DK-5230, Denmark (email:yeo@imada.sdu.dk).
}}

\maketitle

\begin{abstract}
A digraph is {\bf eulerian} if it is connected and every vertex has its in-degree equal to its out-degree. 
 Having a spanning eulerian subdigraph is thus a weakening of having a hamiltonian cycle.
 In this paper, we first characterize the pairs $(D,a)$ of a semicomplete digraph $D$ and an arc $a$ such that $D$ has a spanning eulerian subdigraph containing $a$. In particular, we show that if $D$ is $2$-arc-strong, then every arc is contained in a spanning eulerian subdigraph.
 We then characterize the pairs $(D,a)$  of a semicomplete digraph $D$ and an arc $a$ such that $D$ has a spanning eulerian subdigraph avoiding $a$. In particular, we prove that every $2$-arc-strong semicomplete digraph has a spanning eulerian subdigraph avoiding any prescribed arc. We also prove the existence of a (minimum) function $f(k)$ such that every $f(k)$-arc-strong semicomplete digraph contains a spanning eulerian subdigraph avoiding any prescribed set of $k$ arcs: we prove $f(k)\leq (k+1)^2/4 +1$, 
conjecture $f(k)=k+1$ and establish this conjecture for $k\leq 3$ and when the $k$ arcs that we delete form a forest of stars.
 
 A digraph  $D$ is {\bf eulerian-connected} if for any two distinct vertices $x,y$, the digraph $D$ has a spanning $(x,y)$-trail. We prove that every $2$-arc-strong semicomplete digraph is eulerian-connected.
 
 All our results may be seen as arc analogues of well-known results on hamiltonian cycles in semicomplete digraphs.

\noindent{}{\bf Keywords:} Arc-connectivity, Eulerian subdigraph, Tournament, Semicomplete digraph, polynomial algorithm.
\end{abstract}

\section{Introduction}
%%%%%%%%%%%%

A digraph is {\bf semicomplete} if it has no pair of non-adjacent vertices. A {\bf tournament} is a semicomplete digraph without directed cycles of length 2. Two of the classical results on digraphs are  Camion's Theorem and Red\'ei's theorem (both were originally formulated only for tournaments but they  easily extend to semicomplete digraphs). 

\begin{theorem}[Camion~\cite{Cami59}]\label{thm:camion}
Every strong semicomplete digraph has a hamiltonian cycle.
\end{theorem}

\begin{theorem}[R\'edei~\cite{Rede34}]\label{thm:redei}
  Every semicomplete digraph has a hamiltonian path.
  \end{theorem}

  Thomassen~\cite{thomassenJCT28} proved the following. (It was originally formulated only for tournaments but the proof works for semicomplete digraphs as it easily follows from Theorem~\ref{thm:cond}.) 
\begin{theorem}[Thomassen~\cite{thomassenJCT28}]\label{thm:arc-in-ham}
In a $3$-strong semicomplete digraph, every arc is contained in a  hamiltonian cycle.
\end{theorem}
The $3$-strong assumption in this theorem best possible:  Thomassen~\cite{thomassenJCT28} described an infinite class of $2$-strong tournaments containing an arc which is not in any hamiltonian cycle. It is easy to modify his example to show that there is no $k$ such that every $k$-arc-strong tournament has a hamiltonian cycle containing any given arc.
No characterization of the set of arcs which belong to a hamiltonian cycle in a semicomplete digraph (or a tournament) is known. \\

A natural question is whether the $3$-strong assumption of Theorem~\ref{thm:arc-in-ham} can be relaxed if instead of a hamiltonian cycle, we only require a spanning eulerian subdigraph.
In this paper we answer this question by proving  the following analogue to Theorem~\ref{thm:arc-in-ham}.
% every 2-arc-strong semicomplete digraph contains a spanning eulerian subdigraph which contains any prescribed arc.
\begin{theorem}
  \label{thm:2as-always-good}
  Let $D=(V,A)$ be a $2$-arc-strong semicomplete digraph. For every arc $a\in A$ there exists a spanning eulerian subdigraph of $D$ containing $a$.
  \end{theorem}
 In addition (and contrary to the lack of a known characterization for hamiltonian cycles mentioned above), in Section~\ref{sec:charac}, we characterize the pairs $(D,a)$ such that $D$ is a strong semicomplete digraph containing the arc $a$ and no spanning eulerian subdigraph of $D$ contains the arc $a$.

In Section~\ref{sec:avoid}, we also study spanning eulerian subdigraphs of a semicomplete digraph avoiding a prescribed set of arcs.
Fraisse and Thomassen~\cite{fraisseGC3} proved the following result on hamiltonian cycles avoiding a set of prescribed arcs. For a strengthening of this result, see \cite{bangCPC6}. The connectivity requirement of Theorem~\ref{HCavoidkarcs} is best possible as there are $k$-strong tournaments with vertices of out-degree exactly $k$.

\begin{theorem}[Fraisse and Thomassen~\cite{fraisseGC3}]\label{HCavoidkarcs}
Every $(k+1)$-strong tournament contains a hamiltonian cycle avoiding any prescribed set of $k$ arcs.
\end{theorem}
This theorem does not extend to semicomplete digraphs. Indeed the 2-strong semicomplete digraph obtained from a $4$-cycle by adding a 2-cycle between each of the two pairs of non-adjacent vertices has a unique hamiltonian cycle, and thus no arc of this cycle cannot be avoided.
Observe however that Theorem~\ref{thm:arc-in-ham} implies that every $3$-strong tournament contains a hamiltonian cycle avoiding any prescribed arc.
Improving a previous bound by Bang-Jensen and Thomassen, Guo~\cite{guoDAM79b} proved that every $(3k+1)$-strong semicomplete digraph contains a spanning $(k+1)$-strong tournament. 
Together with Theorem \ref{HCavoidkarcs}, this implies that every $(3k+1)$-strong  semicomplete digraph contains a hamiltonian cycle avoiding any prescribed set
 of $k$ arcs. We conjecture that a much lower connectivity suffices.

\begin{conjecture}\label{conj:avoid}
Let $k$ be a non-negative integer.
Every $(k+2)$-strong semicomplete digraph contains a hamiltonian cycle avoiding any prescribed set of $k$ arcs.
\end{conjecture}
Bang-Jensen and Jord\'an \cite{bangDM310} proved that every 3-strong semicomplete digraph contains a spanning 2-strong tournament. Combining this with Theorem \ref{HCavoidkarcs} shows that Conjecture~\ref{conj:avoid} holds for $k=1$.

\medskip

As an analogue to Theorem \ref{HCavoidkarcs},  we prove  that there is a function $f(k)$ such that every $f(k)$-arc-strong semicomplete digraph contains a spanning eulerian subdigraph avoiding any prescribed set of $k$ arcs. In Theorem~\ref{thm:f(k)}, we show that $f(k)\leq  (k+1)^2/4 +1$. This upper bound is certainly not tight. 
Since there are $k$-arc-strong semicomplete digraphs in which one or more vertices have out-degree $k$, we have $f(k)\geq k+1$. We conjecture that $f(k)=k+1$.
\begin{conjecture}\label{Euleravoidkarcs}
% For every $k$-arc-strong semicomplete digraph $D=(V,A)$ and any set $A'\subseteq A$ of at most $k-1$ arcs, $D\setminus A'$ %has a spanning eulerian subdigraph.
For every non-negative integer  $k$, every $(k+1)$-arc-strong semicomplete digraph $D$ has a spanning eulerian subdigraph that avoids any prescribed set of $k$ arcs.
 \end{conjecture}
Observe that Camion's Theorem implies this conjecture when $k=0$, that is $f(0)=1$.
In Corollary~\ref{cor:k<=2},  we prove Conjecture~\ref{Euleravoidkarcs} for $k\leq 2$ and in Theorem \ref{thm:k=3}, we prove it for $k=3$.
Hence $f(1)=2$, $f(2)=3$ and $f(3)=4$.

In Section~\ref{sec:unavoid}, we characterize the pairs $(D,a)$ such that $D=(V,A)$ is a strong semicomplete digraph, $a\in A$ and every spanning eulerian subdigraph of $D$ contains the arc $a$ (Theorem~\ref{thm:charac-unavoid}).

\medskip

A digraph $D$ is (strongly) {\bf hamiltonian-connected} if for any pair of distinct vertices $x,y$, $D$ has a hamiltonian path from $x$ to $y$.
\iffalse Equivalently, $D$ is {\bf hamiltonian-connected} if for any two vertices $x,y$, the digraph $D\cup \{yx\}$ has a hamiltonian cycle containing the arc $yx$.\fi
Thomassen~\cite{thomassenJCT28} proved the following. (Again it was originally formulated only for tournaments but the proof works for semicomplete digraphs as it easily follows from Theorem~\ref{thm:cond}.)

\begin{theorem}[Thomassen~\cite{thomassenJCT28}]\label{thm:ham-conn}
Every $4$-strong semicomplete digraph is hamiltonian-connected. 
\end{theorem}

The $4$-strong assumption in this theorem best possible:  Thomassen~\cite{thomassenJCT28} described infinitely many $3$-strong tournaments that are not hamiltonian-connected.
Again, it is natural to ask whether the connectivity assumption of Theorem~\ref{thm:ham-conn} can be relaxed if instead of hamiltonian-connected, we only require the digraph to eulerian-connected.
\iffalse
A digraph $D$ is {\bf eulerian-connected} if for any two vertices $x,y$, the digraph $D\setminus \{yx\}$ has a spanning $(x,y)$-trail. Equivalently, $D$ is {\bf eulerian-connected} if for any pair of distinct vertices $x,y$, the digraph $D\cup \{yx\}$ has a spanning eulerian subdigraph containing the arc $yx$. \fi
A digraph $D$ is {\bf eulerian-connected} if for any two vertices $x,y$, the digraph $D$ has a spanning $(x,y)$-trail.
We prove that every $2$-arc-strong semicomplete digraph is eulerian-connected.

\begin{theorem}\label{thm:2-arc}
Every $2$-arc-strong semicomplete digraph is eulerian-connected.
\end{theorem}

 This theorem can been seen as an analogue of Theorem~\ref{thm:ham-conn}.
The $2$-arc-strong condition is best possible. In Proposition~\ref{prop:exple}, we describe strong tournaments with arbitrarily large in- and out-degrees in which there is an arc contained in no spanning eulerian subdigraph.

To prove Theorems~\ref{thm:arc-in-ham} and~\ref{thm:ham-conn}. Thomassen~\cite{thomassenJCT28} gave the following sufficient condition for a semicomplete digraph  to contain a hamiltonian $(x,y)$-path, which implies both results immediately.

\begin{theorem}[Thomassen~\cite{thomassenJCT28}]\label{thm:cond}
Let $T$ be a $2$-strong semicomplete digraph,  and let $x$ and $y$ be two distinct vertices of $T$. If there are three internally disjoint $(x,y)$-paths of length greater than $1$, then there is a hamiltonian $(x,y)$-path in $D$.
\end{theorem}

To prove our results, we prove a theorem that can be seen as an arc analogue to Theorem~\ref{thm:cond}.
\begin{theorem}\label{thm:spanning-trail}
Let $D$ be a strong semicomplete digraph, and let $x$ and $y$ be two vertices of $D$.
If there are two arc-disjoint $(x,y)$-paths in $D$, then there is a spanning $(x,y)$-trail in $D\setminus \{yx\}$.
\end{theorem}
This theorem directly implies Theorems~\ref{thm:2as-always-good} and \ref{thm:2-arc}.

\section{Terminology}
%%%%%%%%%%%%%%%%%%%%%%%%%

Notation generally follows \cite{bang2018,bang2009}. 
The digraphs have no parallel arcs and no loops.
We denote the vertex set and arc set of a digraph $D$ by $V(D)$ and
$A(D)$, respectively and write $D=(V,A)$ where $V=V(D)$ and
$A=A(D)$. 
A {\bf non-edge} of a digraph is a pair $\{x,y\}$ such that $x$ and $y$ are not adjacent, that is neither $xy$ nor $yx$ are arcs.
Unless otherwise specified, the numbers $n$ and $m$ will always be used to denote the number of vertices, respectively arcs, in the digraph in question.
  We use the
notation $[k]$ for the set of integers $\{1,2,\ldots{},k\}$. 

Let $D=(V,A)$ be a digraph. The subdigraph {\bf induced} by a set  $X\subseteq V$ in a digraph $D$ is denoted by $\induce{D}{X}$.
If $X$ is a set of vertices we denote by $D-X$ the digraph $\induce{D}{V\setminus X}$, and if $A'$ is a set of arcs in $D$, then we denote by $D\setminus A'$ the digraph we obtain by deleting all arcs in $A'$.

When $xy$ is an arc of $D$ we say that $x$ {\bf dominates} $y$ and write $x\dom y$. 
If $x\dom y$ for all $x\in X$ and all $y\in Y$, then we write $X\dom Y$ and we write $X\sdom Y$ when $X\dom Y$ and there is no arc from $Y$ to $X$. For sake of clarity, we abbreviate $\{x\}\dom Y$ to $x \dom Y$. 
 For a
digraph $D=(V,A)$ the {\bf out-degree}, $d^+_D(x)$  (resp. the {\bf in-degree}, $d^-_D(x)$) of a vertex $x\in V$ is the number of arcs of the kind 
$xy$ (resp. $yx$) in $A$. %{where we count parallel arcs.
When $X\subseteq V$ we shall also write $d^+_X(v)$ to denote the number
of arcs $vx$ with $x\in X$.
A {\bf sink} in a digraph is a vertex with out-degree $0$ and a {\bf source} is a vertex with in-degree $0$.

A {\bf walk} is an alternating sequence $W=(v_0, a_1, v_1, \dots , a_p, v_p)$ of vertices and 
arcs such that $a_i=v_{i-1}v_i$ for all $1\leq i\leq p$. 
Its {\bf initial vertex}, denoted by  $s(W)$, is $v_0$ and its {\bf terminal vertex}, denoted by $t(W)$, is $v_p$. The $v_i$, $1\leq i\leq p-1$, are the {\bf internal vertices} of $W$.
A walk is completely determined by the sequence of its vertices. Therefore for the sake of simplicity, we use the sequence
$v_0v_1\cdots v_p$ to denote the walk $(v_0, a_1, v_1, \dots , a_p, v_p)$.

A walk $W$ is {\bf closed} if  $s(W)=t(W)$.
A {\bf trail} is a walk in which all arcs are distinct, a {\bf path} is a walk in which all vertices are distinct and a {\bf cycle} is a closed walk in which all vertices are distinct except the initial and terminal vertices. 
Note that, walks, trails, paths and cycles are always directed.
%\jbj{A closed trail is also called an {\bf eulerian trail}.}

An {\bf $(s,t)$-walk} (resp.  {\bf $(s,t)$-trail}, {\bf $(s,t)$-path} is a walk (resp. trail, path) with initial vertex
$s$ and terminal vertex $t$.  
Observe that if $s\neq t$, then an $(s,t)$-trail can be seen as a connected digraph such that 
$d^+(s) = d^-(s)+1$, $d^-(t) = d^+(t)+1$ and $d^+(v) = d^-(v)$ for all other vertices.
For two sets $X, Y$ of vertices, an $(X,Y)$-path is a path with initial vertex in $X$, terminal vertex in $Y$, and no internal vertices in $X\cup Y$.

Let $P=x_1\cdots x_p$ be a path. For any $1\leq i \leq j\leq p$, we denote by $P[x_i,x_j]$ the path $x_i \cdots x_j$,
by $P[x_i,x_j)$ the path $x_i \cdots  x_{j-1}$, by $P(x_i,x_j]$ the path $x_{i+1} \cdots x_{j}$, and by $P(x_i,x_j)$ the path $x_{i+1} \cdots  x_{j-1}$.
Similarly, if $C$ is a cycle and $x,y$ two vertices of $C$, we denote by $C[x,y]$ the $(x,y)$-path in $C$ if $x\neq y$ and the cycle $C$ if $x=y$. 
Denote by $x^+$ the out-neighbour of $x$ in $C$ and by $y^-$ the in-neighbour of $y$ in $C$, and let
 $C(x,y] = C[x^+,y]$, $C[x,y) = C[x,y^-]$ and $C(x,y) = C[x^+,y^-] $. 

 A vertex $v$ of a digraph $D$ is an {\bf out-generator} (resp. {\bf in-generator}) if $v$ can reach (resp. be reached by) all other vertices by paths.

 \iffalse
 The following well-known fact is easy to prove by induction (see also \cite{bangJGT20a}).
prede
\begin{theorem}\label{out/in-generator}
A semicomplete digraph $D$ has a hamiltonian path starting (ending) in the vertex $x$ if and only if $x$ is an out-generator (in-generator) of $D$.
\end{theorem}
\fi

A digraph $D$ is {\bf eulerian} if it contains a spanning eulerian trail $W$ such that $A(W)=A(D)$, or, equivalently by Euler's theorem, if $D$ is connected and $d^+(v)=d^-(v)$ for all $v\in V(D)$.
%A digraph is {\bf supereulerian} if it contains a closed trail $W$ such that $V(W)=V(D)$, or, equivalently, if it contains a spanning eulerian subdigraph.

The {\bf underlying (multi)graph} of a digraph $D$,
denoted $UG(D)$, is obtained from $D$ by suppressing the orientation
of each arc. 
A digraph $D=(V,A)$ is {\bf
  connected} if $UG(D)$ is a connected graph. 
It is {\bf strong} if it contains an $(s,t)$-path for each ordered pair of distinct vertices $s,t\in V$. 
It is {\bf $\mathbf{k}$-strong} if $D - W$ is strong for every subset $W\subseteq V$ of at most $k-1$ arcs. 
%The largest $k$ such that $D$ is $k$-strong is called the {\bf connectivity} of $D$ and is denoted by $\kappa{}(D)$.
It is {\bf $\mathbf{k}$-arc-strong} if $D\setminus A'$ is strong for every subset $A'\subseteq A$ of at most $k-1$ arcs. 
The largest $k$ such that $D$ is $k$-arc-strong is called the {\bf arc-connectivity} of $D$ and is denoted by $\lambda{}(D)$.
A {\bf cut-arc} in $D$ is an arc $a$ such that $D\setminus a$ is not strong.

An {\bf independent set} in a digraph $D$ is a set of pairwise non-adjacent vertices.
 The {\bf independence number} of $D$, denoted by $\alpha(D)$,  is the maximum size of an independent set in $D$.

\section{Structure of semicomplete digraphs}\label{subsec:structure}
%%%%%%%%%%%%%%%%%%%%%%%%%%%%%%%%%%

Let $D$ be a digraph.
A {\bf decomposition} of $D$ is a partition $(S_1, \dots , S_p)$, $p\geq 1$, of its vertex set.
The {\bf index} of vertex $v$ in the decomposition, denoted by $\ind(v)$, is the integer $i$ such that $v\in S_i$.
An arc $uv$ is {\bf forward} if $\ind(u) < \ind(v)$, {\bf backward} if $\ind(u) > \ind(v)$, and {\bf flat} if $\ind(u) = \ind(v)$.
For sake of clarity, we often abbreviate $S_{\ind(u)}$ into $S_u$.

A  decomposition $(S_1,\ldots{},S_p)$ is {\bf strong} if $D\langle S_i \rangle$ is strong for all $1\leq i\leq p$.
The following proposition is well-known (just consider an acyclic ordering of the strong components of $D$).
\begin{proposition}\label{prop:noback}
Every digraph has a strong decomposition with no backward arcs.
\end{proposition}

A {\bf $1$-decomposition} of a digraph $D$ is a strong decomposition such that every backward arc is a cut-arc and all cut-arcs are either forward or backward.

\begin{proposition}
Every strong digraph admits a $1$-decomposition.
\end{proposition}
\begin{proof}

Let $D$ be a strong digraph and let $C$ be its set of cut-arcs. If $C=\emptyset$, then the trivial decomposition with only one set $S_1=V(D)$ is a 1-decomposition, so assume that $C\neq\emptyset$. Observe that
$D\setminus C$ is not strong. Thus,  by Proposition~\ref{prop:noback},  $D\setminus C$ has a strong decomposition $(S_1, \dots, S_p)$ with no backward arcs.
This decomposition is clearly a $1$-decomposition of $D$.
\end{proof}

\iffalse
\begin{lemma}
\label{disjointcutarcs}
In a semicomplete digraph no two cut-arcs can have the same head or tail.
\end{lemma}
\begin{proof}
If $uv,uw$ ($vu,wu$) are arcs with the same tail (head) and $vw$ is an arc, then there is a $(u,w)$-path in $D\setminus uw$. Thus $uw$ is not a cut-arc.
\end{proof}}
\fi

Let $(S_1, \dots , S_p)$ be a decomposition of a digraph.
Two backward arcs $uv$ and $xy$ are {\bf nested} if either $\ind(v) \leq \ind(y) < \ind(x) \leq \ind(u)$ or $\ind(y) \leq \ind(v) < \ind(u) \leq \ind(x)$.
See Figure~\ref{fig:nested}.
\begin{figure}[hbtp]
\begin{center}
\tikzstyle{vertexX}=[circle,draw, top color=gray!5, bottom color=gray!30, minimum size=16pt, scale=0.6, inner sep=0.5pt]
\tikzstyle{vertexY}=[circle,draw, top color=gray!5, bottom color=gray!30, minimum size=20pt, scale=0.7, inner sep=1.5pt]

%\tikzstyle{vertexBIG}=[ellipse, draw, scale=0.6, inner sep=3.5pt]
\begin{tikzpicture}[scale=0.6]
 \draw (4,0.6) node {$S_1$};
 %\node (x11) at (4.0,2.0) [vertexX] {$x_1$};
 %\node (x12) at (4.0,5.0) [vertexX] {$\bar{x}_1$};
 \draw [rounded corners] (3.3,1.2) rectangle (4.7,5.8);
  \draw (5.5,3.5) node {$\cdots$};
 \draw [rounded corners] (6.3,1.2) rectangle (7.7,5.8);
 \node (v1) at (7.0,5.0) [vertexX] {$v_1$};
 \node (v3) at (7.0,2.0) [vertexX] {$v_3$};
  \draw (8.5,3.5) node {$\cdots$};
 \node (v2) at (10.0,3.5) [vertexX] {$v_2$};
 \draw [rounded corners] (9.3,1.2) rectangle (10.7,5.8);
  \draw (11.5,3.5) node {$\cdots$};
 \draw [rounded corners] (12.3,1.2) rectangle (13.7,5.8);
 \node (u3) at (13.0,2.0) [vertexX] {$u_3$};
  \draw (14.5,3.5) node {$\cdots$};
 %\draw (16,0.6) node {$X_4$};
  \draw [rounded corners] (15.3,1.2) rectangle (16.7,5.8);
  \node (u2) at (16.0,3.5) [vertexX] {$u_2$};
  \draw (17.5,3.5) node {$\cdots$};
 \draw [rounded corners] (18.3,1.2) rectangle (19.7,5.8);
 \node (u1) at (19.0,5.0) [vertexX] {$u_1$};
 \draw (20.5,3.5) node {$\cdots$};
 \draw (22,0.6) node {$S_{p}$};
 \draw [rounded corners] (21.3,1.2) rectangle (22.7,5.8);

\draw [->, line width=0.03cm] (u1) to [out=155, in=25] (v1);
\draw [->, line width=0.03cm] (u2) to [out=160, in=20] (v2);
\draw [->, line width=0.03cm] (u3) to [out=160, in=20] (v3);

\end{tikzpicture}
\end{center}
\caption{Illustration of nested backwards arcs. The arcs $u_1v_1$ and $u_2v_2$ are nested;  the arcs $u_1v_1$ and $u_3v_3$ are nested; the arcs $u_2v_2$ and $u_3v_3$ are not nested.}\label{fig:nested}
\label{fig2}
\end{figure}
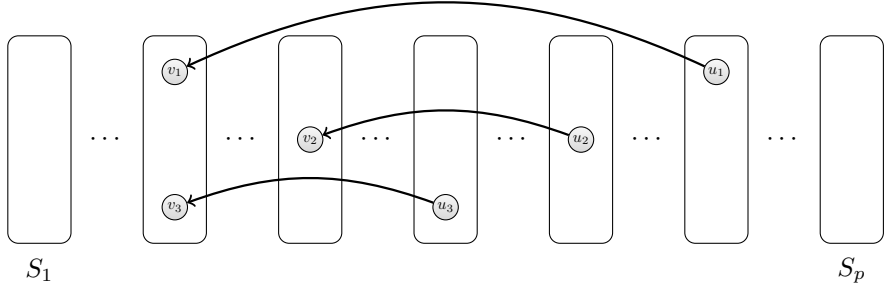

\begin{proposition}\label{prop:nosame}
Let $(S_1, \dots, S_p)$ be a $1$-decomposition of a strong semicomplete digraph $D$.
The following properties hold:
\begin{itemize}
 \item[(i)] If $u_1v_1$ and $u_2v_2$ are two cut-arcs, then $\ind(u_1)\neq \ind(u_2)$ and $\ind(v_1)\neq \ind(v_2)$.
\item[(ii)] There are no nested backward arcs.
\item[(iii)] If $|V(D)|\geq 4$ and $uv$ is a forward cut-arc, then $|S_u|=|S_v|=1$ and $\ind(v) = \ind(u)+1$.
\end{itemize}
\end{proposition}
\begin{proof}

(i) Assume for a contradiction that $\ind(u_1) = \ind(u_2)$. Since $D$ is semicomplete, there is an arc between $v_1$ and $v_2$.
Without loss of generality, we may assume that $v_1v_2$ is an arc. In $D\langle S_{u_1}\rangle = D\langle S_{u_2}\rangle$, there is a $(u_2,u_1)$-path $P$. Note that $P$ avoids $u_2v_2$ because this arc is not flat. But then $P\cup u_1v_1v_2$ is a $(u_2, v_2)$-path in $D\setminus u_2v_2$, contradicting that   $u_2v_2$ is  a cut-arc.

\medskip

(ii) Suppose for a contradiction that $D$ contains two nested arcs $uv$ and $xy$ such that $\ind(v) \leq \ind(y) < \ind(x) \leq \ind(u)$.
By (i), $\ind(v) < \ind(y)$ and $\ind(x) < \ind(u)$. 
Moreover by (i), 
$D$ contains the arcs $vy,xu$. But now $xuvy$ is an $(x,y)$-path in $D\setminus xy$, contradicting the fact that $xy$ is a cut-arc.

\iffalse
Without loss of generality, we may assume that $y$ has the lowest possible index, that is, there is no arc $x'y'$ such that $\ind(v) \leq \ind(y') < \ind(x') \leq \ind(u)$, and $\ind(y') < \ind(y)$.
Hence $uy$ is a forward arc. Therefore $xv$ is not an arc for otherwise $(x,v,u,y)$ would be an $(x,y)$-path in $D\setminus xy$ contradicting the fact that $xy$ is a cut-arc. Hence $vx$ is an arc. It is backward and has same tail as $(u,v)$, a contradiction to (i).
\fi

\medskip

(iii) Assume $|D|\geq 4$ and let $uv$ be a forward cut-arc.

For any vertex $u'$ in $S_u\setminus \{u\}$, there is a $(u,u')$-path $P$ in $D\langle S_u\rangle$, and so $vu'$ is a backward arc for otherwise  $P\cup u'v$ would be a $(u,v)$-path in $D\setminus uv$. Hence by (i), $|S_u\setminus \{u\}| \leq 1$, so $|S_u|\leq 2$.

Assume for a contradiction that $|S_u|=2$, say $S_u=\{u,u'\}$. 
Let $S=S_{ind(u)+1}\cup\cdots\cup{}S_v$.  If $v$ has an in-neighbour $w$ in $S$ then, by (i), $uw$ is an arc (since $vu'$ is a backward arc, and so $uwv$ is a $(u,v)$-path, a contradiction to the fact that  $uv$ is a cut-arc. Hence, by (i), $S=\{v\}$.
Now since $|V(D)|\geq 4$, either $\ind(u)>1$ or $\ind(v)<p$. By (ii) $vu'$ is the only  arc from $S_v\cup\cdots{}\cup S_p$ to $S_1\cup\cdots\cup{}S_u$, and by (i) the only cut-arc with tail in $S_u$ is $uv$, and the only cut-arc with head in $S_v$ is $uv$. Therefore, if $\ind(u)>1$, there is no arc from $S_u\cup\cdots{}\cup S_p$ to  $S_1\cup\cdots\cup{}S_{\ind(u)-1}$,
and if $\ind(v)<p$, there is no arc  from $S_{\ind(v)+1}\cup\cdots{}\cup S_p$ to $S_1\cup\cdots\cup{}S_v$. This is a contradiction to the fact that $D$ is strong.

Hence $|S_u|=1$. Symmetrically, we obtain $|S_v|=1$.

Let $W=\{w \mid \ind(u) < \ind(w) < \ind(v)\}$, $X=\{x \mid \ind(x)<\ind(u)\}$, and $Y= \{y \mid \ind(v) < \ind(y)\}$.
Observe that for every $w\in W$, either $uw\notin A(D)$ or $wv\notin A(D)$ for otherwise $uwv$ would be a $(u,v)$-path in $D\setminus uv$ 
(contradicting that $uv$ is a cut-arc).
Since $D$ is semicomplete, this implies that one of the two arcs $wu$, $vw$  is a backward arc.
In particular, $|W|\leq 2$ for otherwise either there would be two backward arcs with tail $v$ or two backwards arcs with head $u$, contradicting (i).

Assume for a contradiction that $|W|=2$, say $W=\{w_1,w_2\}$ and $w_1\dom w_2$.
If $uw_1$ is an arc then the fact that $uv$ is a cut-arc would imply that $v$ would have backwards arcs to each of $w_1,w_2$, contradicting~(i). Hence $uw_1$ is not an arc and  $D$ contains the arcs $w_1u$ (as $u w_1 \not\in A(D)$), $uw_2$ (by (i)), $vw_2$ (as $uv$ is a cut-arc) and $w_1v$ (by (i))
and does not contain the arcs $w_2u, vw_1,w_2w_1$. Observe that by (i) $w_1w_2$ is not a cut-arc and so $\ind(w_2)\geq \ind(w_1)$.
Since $D$ is strong, $w_1$ must have an in-neighbour $z$, which must be in $X\cup Y$.
If $X\neq \emptyset$, then there must be an arc from $W\cup Y\cup\{u,v\}$ to $X$. By  (i) the tail of this arc is not in $\{u, w_1\}$ and so this arc and $w_1u$ are two nested backward arcs, a contradiction to (ii).  Similarly, we get a contradiction if $Y\neq \emptyset$.
However $X=\emptyset$ and $Y=\emptyset$ is a contradiction to $z \in X \cup Y$.

Assume for a contradiction that $|W|=1$, say $W=\{w\}$.
Since $uv$ is a cut-arc, then $uwv$ cannot be a path, so either $uw$ or $wv$ is not an arc.

Let us assume that $uw$ is not an arc. Then $wu\in A(D)$ because $D$ is semicomplete.
Thus $X=\emptyset$, for otherwise $wu$ and any arc  from $Y\cup\{u,v,w\}$ to $X$ would be two nested arcs 
(as by (i) it can not leave $\{u\}$), a contradiction to (ii).
Hence $Y\neq \emptyset$, since $|D|\geq 4$. So there must be an arc from $Y$ to $\{u,v,w\}$. By (i), the head of this arc must be $w$.
Let $y$ be its tail. By (i) $vw$ and $yu$ are not backward arcs, so $uywv$ is a $(u,v)$-path in $D\setminus uv$, a contradiction.

Similarly, we get a contradiction if $wv$ is not an arc. Hence $W=\emptyset$, that is $\ind(v) = \ind(u)+1$.
\end{proof}

A {\bf nice decomposition} of a digraph $D$ is a $1$-decomposition such that the set of cut-arcs of $D$ is exactly the set of backward arcs.

\begin{proposition}\label{prop:nice}
Every strong semicomplete digraph of order at least $4$ admits a nice decomposition.
\end{proposition}
\begin{proof}
Let $D$ be a  strong semicomplete digraph of order at least $4$. 
If $uv$ has a cut-arc, which is forward.
By Proposition~\ref{prop:nosame} (iii), $S_u=\{u\}$, $S_v=\{v\}$, and $\ind(v) = \ind(u) +1$.
Inverting $S_u$ and $S_v$ (that is, considering the decomposition $(S_1, \dots, S_{\ind(u)-1}, \{v\}, \{u\} , S_{\ind(u)+2}, \dots , S_p)$ ), we obtain another $1$-decomposition with one forward cut-arc less.
Doing this for all forward cut-arcs, we obtain a nice decomposition of $D$.
\end{proof}

Given a semicomplete digraph and a nice decomposition of it, the {\bf natural ordering} of its backward arcs is the ordering in decreasing order according to the index of their tail.
Note that this ordering is unique by Proposition~\ref{prop:nosame}~(i).

\begin{proposition}\label{prop:ordering}
Let $D$ be a  strong semicomplete digraph of order at least $4$, let $(S_1, \dots , S_p)$ be a nice decomposition of $D$, and let $(s_1t_1, s_2t_2, \dots , s_rt_r)$ be the natural ordering of the backward arcs. Then
\begin{itemize}
\item[(i)] %$\ind(t_{j+2}) < \ind(t_{j+1})\leq \ind(s_{j+2}) < \ind(t_j)\leq \ind(s_{j+1}) < \ind(s_j)$ for all $1\leq j \leq r-2$;
$\ \ind(t_{j+1}) < \ind(t_j)\leq \ind(s_{j+1}) < \ind(s_j)$ for all $1\leq j \leq r-1$ and

$\ind(t_{j+1})\leq \ind(s_{j+2}) < \ind(t_j)$ for all $1\leq j \leq r-2$;

\item[(ii)] $s_1\in S_p$ and $t_r\in S_1$;

\item[(iii)] If  $\ind(t_j) = \ind(s_{j+1})=i$ and $t_j\neq s_{j+1}$, then there are two arc-disjoint $(t_j, s_{j+1})$-paths in $D\langle S_i \rangle$.
 \end{itemize}
\end{proposition}
\begin{proof}
(i) By Proposition~\ref{prop:nosame}~(i), $\ind(s_{j+1}) < \ind(s_j)$, and as $D$ is strong, $\ind(t_j)\leq \ind(s_{j+1}) < \ind(s_j)$.
By Proposition~\ref{prop:nosame}~(ii), $s_{j}t_{j}$ and $s_{j+1}t_{j+1}$ are not nested so $\ind(t_{j+1}) <  \ind( t_{})$.
Assume for a contradiction that $\ind(t_j) \leq \ind(s_{j+2})$. By Proposition~\ref{prop:nosame} (i), $s_{j} s_{j+1}$ and $t_{j+1}t_{j+2}$ are not arcs, so $s_{j+1} s_{j}$ and $t_{j+2}t_{j+1}$ are arcs. 
If $\ind(t_{j}) <  \ind( s_{j+2})$, then $t_{j}s_{j+2}\in A(D)$, and if $\ind(t_{j}) =  \ind( s_{j+1})$, then there is a $(t_j, s_{j+2})$-path in $D\langle S_{t_j}\rangle$. In both cases, there is a $(t_j, s_{j+2})$-path $P$ not using the arc $s_{j+1}t_{j+1}$. Now $s_{j+1}s_jt_j\cup P \cup s_{j+2}t_{j+2}t_{j+1}$ is an $(s_{j+1}, t_{j+1})$-path in $D\setminus s_{j+1}t_{j+1}$, a contradiction.

\medskip

(ii) Because $D$ is strong, there must be a backward arc with tail in $S_p$ and a backward arc with head in $S_1$.
By the above inequality, necessarily $s_1\in S_p$ and $t_r\in S_1$.

\medskip

(iii) Assume for a contradiction that $\ind(t_j) = \ind(s_{j+1})=i$ and there do not exist two arc-disjoint $(t_j, s_{j+1})$-paths in $D\langle S_i \rangle$. By Menger's Theorem, there is an arc $a$ such that $D\langle S_i \rangle \setminus \{a\}$ has no $(t_j, s_{j+1})$-path.
But then, there is no  $(t_j, s_{j+1})$-path in $D \setminus \{a\}$, that is $a$ is a cut-arc of $D$. This contradicts the fact that $(S_1, \dots , S_p)$ is a nice decomposition.
\end{proof}

\section{Eulerian-connected semicomplete digraphs}\label{sec:spanning}
%%%%%%%%%%%%%%%%%%%%%%%%%%%%%%%%%%%%

We first observe that being strong and having large in- and out-degrees is not sufficient to guarantee every arc of a tournament to be in a spanning eulerian subdigraph.

\begin{figure}[!hbtp]
\begin{center}
\tikzstyle{vertexX}=[circle,draw, top color=gray!5, bottom color=gray!30, minimum size=16pt, scale=0.7, inner sep=0.5pt]
\tikzstyle{vertexY}=[circle,draw, top color=gray!5, bottom color=gray!30, minimum size=55pt, scale=1.5, inner sep=1.5pt]
\tikzstyle{vertexZ}=[circle,draw, top color=gray!0, bottom color=gray!0, minimum size=60pt, scale=1.0, inner sep=1.5pt]

\begin{tikzpicture}[scale=0.4]
 \node (QQ) at (8.0,9.0) [vertexZ] {};
 \node (AA) at (3.0,4.0) [vertexY] {$A$};
 \node (BB) at (13.0,4.0) [vertexY] {$B$};
 \node (x) at (6.6,8) [vertexX] {$x$};
 \node (y) at (8.0,10) [vertexX] {$y$};
 \node (z) at (9.4,8) [vertexX] {$z$};
 \node (a) at (3.9,2.6) [vertexX] {$a$};
 \node (b) at (12.1,2.6) [vertexX] {$b$};

\draw [->, line width=0.03cm] (x) -- (y);
\draw [->, line width=0.03cm] (y) -- (z);
\draw [->, line width=0.03cm] (x) -- (z);
\draw [->, line width=0.03cm] (AA) -- (QQ);
\draw [->, line width=0.03cm] (AA) -- (BB);
\draw [->, line width=0.03cm] (QQ) -- (BB);
\draw [->, line width=0.03cm] (b) -- (a);
\end{tikzpicture}
\caption{The tournament $T$ in Proposition~\ref{prop:exple}.}\label{picEXPLE}
\end{center}
\end{figure}
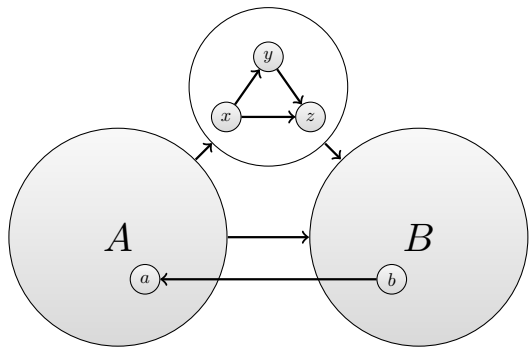

\begin{proposition}\label{prop:exple}
For every positive integer $k$, there exist strong tournaments with minimum in- and out-degrees at least $k$ containing an arc which is not in any spanning eulerian subdigraph.
\end{proposition}
\begin{proof}
Let $T$ (see Figure~\ref{picEXPLE}) be a  tournament  with vertex set $A\cup B \cup \{x,y,z\}$ such that $A\dom \{x,y,z\}$, $\{x,y,z\}\dom B$, $x\dom \{y,z\}$, $y\dom z$, there exist a vertex $a\in A$ and a vertex $b\in B$ such that $T$ contains all arcs from $A$ to $B$ except $ab$ (and so $b\dom a$), and $T\langle A\rangle$ and $T\langle B\rangle$ are strong tournaments with minimum in- and out-degrees at least $k$.
Clearly $T$ is strong and has minimum in- and out-degrees at least $k$. 
One can check that every eulerian subdigraph  containing the arc $xz$ does not contain $y$ and is therefore not spanning.
\end{proof}

In the remaining of the section, we prove Theorem~\ref{thm:spanning-trail}, which we recall.

\medskip
\noindent
{\bf Theorem~\ref{thm:spanning-trail}. }
{\em Let $D$ be a strong semicomplete digraph, and let $x$ and $y$ be two vertices of $D$.
If there are two arc-disjoint $(x,y)$-paths in $D$, then there is a spanning $(x,y)$-trail in $D\setminus \{yx\}$.
}

\medskip

Let us start we some useful preliminaries.
The following lemma is easy and well-known.
\begin{lemma}\label{lem:path-reduc}
Let $D$ be a non-strong semicomplete digraph.
For every out-generator $x$ of $D$ and in-generator $y$ of $D$, there is a hamiltonian $(x,y)$-path in $D$.
\end{lemma}

Lemma~\ref{lem:path-reduc} and Camion's Theorem immediately imply the following.

\begin{corollary}\label{cor:origin}
In a semicomplete digraph,
every out-generator is the initial vertex of a hamiltonian path.
\end{corollary}

We shall now prove a lemma which is a strengthening of Camion's Theorem.

\begin{lemma}\label{lem:camion-extended}
Let $D$ be a semicomplete digraph, $F$ a subdigraph of $D$, and $z$ a vertex in $V(F)$.
If $D\setminus A(F)$ is strong, then there is a cycle containing all vertices of $V(D) \setminus V(F)$ and $z$. 
\end{lemma}
\begin{proof}
Let $D'= D\langle (V(D) \setminus V(F)) \cup \{z\}\rangle$.
If $D'$ is strong, then by Camion's Theorem, it has a hamiltonian cycle, which has the desired property.

If $D'$ is not strong, then let $X$ be its set of out-generators and let $Y$ be its set of in-generators. Since $D\setminus A(F)$ is strong, there is a $(Y,X)$-path $P$ in $D$. Set $D''=D' - P(s(P), t(P))$.
Clearly, $t(P)$ is an out-generator of $D''$ and $s(P)$ is an in-generator of $D''$.
Hence, by Lemma~\ref{lem:path-reduc}, $D''$ has a hamiltonian path $Q$ from $t(P)$ to $s(P)$.
The union of $P$ and $Q$ is the desired cycle.
\end{proof}

\begin{proof}[Proof of Theorem~\ref{thm:spanning-trail}]
We proceed by induction on the number of vertices, the result holding trivially when $|V(D)|=3$.

By the assumption there are two arc-disjoint $(x,y)$-paths $P_1, P_2$. Let $y'_i$ be the out-neighbour of $x$ in $P_i$ and let $x'_i$ be the in-neighbour of $y$ in $P_i$.
We assume that $P_1\cup P_2$ has as few arcs as possible and under this assumption that $P_1$ is as short as possible.
In particular, $x'_2$ and $y'_2$ are not in $V(P_1)$ and all internal vertices of $P_1$ except $y'_1$ dominate $x$, and 
all internal vertices of $P_1$ except $x'_1$ are dominated by $y$.

\bigskip

Assume first that $x\dom y$. %Then $D\setminus \{yx\} =D$. 
By our choice of $P_1$ and $P_2$, we have $P_1=xy$.
The digraph $D\setminus A(P_1)$ is $D\setminus \{yx\}$ and contains $P_2$. Hence it is strong, so by Lemma~\ref{lem:camion-extended}, $D\setminus A(P_1)$ contains a cycle $C$ covering all vertices of $V(D) \setminus \{y\}$.
The union of $C$ and $P_1$ is a spanning $(x,y)$-trail in $D\setminus \{yx\}$.

\bigskip

Assume now that $xy\notin A(D)$. Then $y\dom x$ and $P_1$ has length at least $2$.
Let $w_1$ be the in-neighbour of $x'_1$ on $P_1$.
Set $D'= D\setminus \{yx\}$.

Assume first that $D'$ is not strong. Since $D$ is strong, by Camion's Theorem, it contains a hamiltonian cycle $C$.
Now $C$ must contain the arc $yx$, and $C\setminus \{yx\}$ is a hamiltonian $(x,y)$-path, and so a spanning $(x,y)$-trail in $D'$.
Henceforth, we assume that $D'$ is strong.
%Let $Q_1 = P_1$ if $D'=D$ and $Q_1=P_1\cup \{yx\}$ otherwise. Observe that $Q_1$ is either or path or a cycle.

\medskip

If $D'\setminus A(P_1)$ is strong, then,
by Lemma~\ref{lem:camion-extended}, $D'\setminus A(P_1)$ contains a cycle $C$ covering all vertices of $V(D) \setminus V(P_1)$ and a vertex of $V(P_1)$.
The union of $C$ and $P_1$ is a spanning $(x,y)$-trail in $D'$.
Henceforth we may assume that $D'\setminus A(P_1)$ is not strong. 
Let $(X,Y)$ be a partition of $V(D)$ such that there is no arc from $Y$ to $X$ in $D'\setminus A(P_1)$ and $Y$ is minimal with respect to inclusion. Then it is easy to see that $D\langle Y \rangle$ is strong.
Since $D'$ is strong, there must be an arc of $P_1$ with tail in $Y$ and head in $X$.
Observe that because $P_2$ is a path in $D'\setminus A(P_1)$, we cannot have $x\in Y$ and $y\in X$.

Assume for a contradiction that $x\in X$ and $y\in Y$.
The vertex $x'_1$ is the unique vertex of $P_1$ in $X$ because all other internal vertices of $P_1$ are dominated by $y$.
Similarly, vertex $y'_1$ is the unique vertex of $P_1$ in $Y$ because all others internal vertices of $P_1$ dominate $x$.
 So $P_1=xy'_1x'_1y$.
 Consider now $P_2$ and recall that $x'_2,y'_2\not\in V(P_1)$ and  $|V(P_2)|\geq |V(P_1)|=4$. The vertex $y'_2$ is dominated by $y$, so it must be in $Y$. Similarly, $x'_2$ dominates $x$, so it must be in $X$.
 But then an arc of $A(P_2)$ must have tail in $Y$ and head in $X$, a contradiction.

Assume that $x,y\in Y$.
Vertex $x'_1$ is the unique vertex of $P_1$ in $X$ because all other internal vertices of $P_1$ are dominated by $y$. Furthermore $w_1x'_1$ is the unique arc of $D$ from $Y$ to $X$.
Moreover, since $D$ is strong, $x'_1$ must be an out-generator of $D\langle X \rangle$. Thus, by Corollary~\ref{cor:origin}, there is a hamiltonian path $Q_X$ of $D\langle X \rangle$ with initial vertex $x'_1$.
%Observe now that $x$ and $y$ must be in the same strong component $S$ of $D\langle Y\rangle$ because $P_2\cup \{yx\}$ is a cycle in this digraph. Moreover, because $w_1x'_1$ is the unique arc of $D$ from $Y$ to $X$, $S$ is the terminal strong component of $D\langle Y\rangle$, that is $D\langle Y\rangle - S \dom S$.
%Let $Q_R$ be a hamiltonian path of $D\langle Y\rangle - S \dom S$. (If $D\langle Y\rangle = S$, then it is an empty path).
%Let $Q = Q_X$ if $Q_R$ is empty, and $Q=Q_X\cup(t(Q_X), s(Q_R)) \cup Q_R$ otherwise.
%The initial vertex of $Q_X$ is $x'_1$, which is dominated by $w_1$, and 
The terminal vertex of $Q_X$ dominates $Y\setminus \{w_1\}$.
Let $D'' = D\langle Y\rangle \cup \{w_1y\}$. This digraph is strong. Observe moreover that $w_1y$ was not in $A(D)$ by our choice of $P_1$. Therefore $P_1[x,w_1]\cup w_1y$ and $P_2$ are two arc-disjoint $(x,y)$-paths in $D''$.
By the induction hypothesis, there is a spanning $(x,y)$-trail $W$ in $D''$.
Let $u$ be an out-neighbour of $w_1$ in $W$. Replacing the arc $w_1u$ by $wx'_1\cup Q_X \cup t(Q_X)u$, we obtain a spanning $(x,y)$-trail in $D$.

By symmetry, we get the result if $x,y\in X$.
\end{proof}

\begin{remark}\ 

  \begin{itemize}
\item Note that in the spanning $(x,y)$-trail given by the above proof, every vertex has out-degree at most $2$.

\item The proof of Theorem~\ref{thm:spanning-trail} can easily be translated into a polynomial-time algorithm.
\end{itemize}
\end{remark}

\section{Arcs contained in no spanning eulerian subdigraph}\label{sec:charac}
%%%%%%%%%%%%%%%%%%%%%%%%%%%%%%%

The aim of this section is to prove a characterization of the arcs of a semicomplete digraph $D$  that are not contained in any spanning eulerian subdigraph of $D$. Observe that if the semicomplete digraph is not strong, then there are only such arcs, and if the semicomplete digraph is $2$-strong there are no such arcs by Theorem~\ref{thm:2as-always-good}.

We first deal with digraphs of order at most $3$, before settling the case of digraphs of order at least $4$, for which we use structural properties established in Subsection~\ref{subsec:structure}.

\medskip

Let $D_3$ be the digraph with vertex set $\{x,y,z\}$ and arc set $\{xy, yz, zy , zx\}$.
The following easy proposition is left to the reader.
\begin{proposition}\label{prop:small}
Let $D$ be a strong semicomplete digraph $D$ of order at most $3$ and let $a$ be an arc of $D$.
The arc $a$ is contained in a spanning eulerian subdigraph unless $D=D_3$ and $a=zy$.
\end{proposition}

Let $D$ be a  strong semicomplete digraph of order at least $4$, $(S_1, \dots , S_p)$ a nice decomposition of $D$, and $(s_1t_1, s_2t_2, \dots , s_rt_r)$ the natural ordering of the backward arcs. A set $S_i$ is {\bf ignored} if there exists $j$ such that $\ind(s_{j+1})< i < \ind(t_{j-1})$ or $1 < i <\ind(t_{r-1})$ or $\ind(s_{2})< i < p$. 
An arc $uv$ of $D$ is {\bf regular-bad} if it is forward and there is an integer $i$ such that $\ind(u) < i < \ind(v)$ and $S_i$ is ignored (see Figure \ref{fig:ignored}.)
%It is {\bf right-bad} if $S_{p-2}=\{u\}$, $S_{p-1}=\{v\}$, and $S_p=\{w\}$ for some vertex $w$, and $A(D\langle \{u,v,w\}\rangle) = \{uv, uw, vu, wv\}$.
%It is {\bf left-bad} if $S_{2}=\{u\}$, $S_{3}=\{v\}$, and $S_1=\{w\}$ for some vertex $w$, and $A(D\langle \{u,v,w\}\rangle) = \{uv, vu, uw,wv\}$.
The arc $uv$ is {\bf left-bad} if $S_{2}=\{u\}$, $S_1=\{t_r\}$, $t_r\neq v$, and $t_ru \notin A(D)$.
The arc $uv$ is {\bf right-bad} if $S_{p-1}=\{v\}$, $S_p=\{s_1\}$, $s_1\neq v$, and $us_1 \notin A(D)$.
An arc is {\bf bad} if it is regular-bad, right-bad or left-bad.
A non-bad arc is {\bf good}.

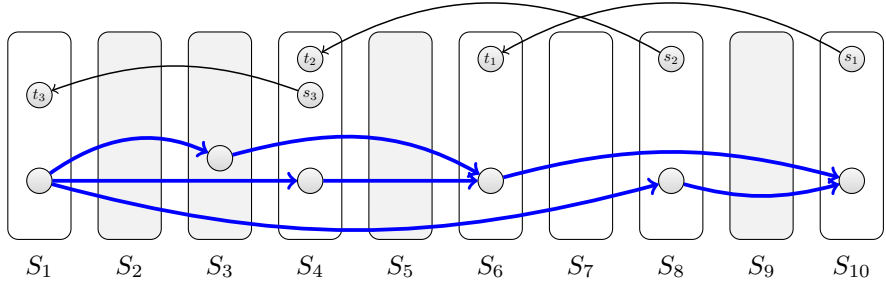
\begin{figure}[hbtp]
\begin{center}
\tikzstyle{vertexX}=[circle,draw, top color=gray!5, bottom color=gray!30, minimum size=16pt, scale=0.6, inner sep=0.5pt]
\tikzstyle{vertexY}=[circle,draw, top color=gray!5, bottom color=gray!30, minimum size=20pt, scale=0.7, inner sep=1.5pt]

%\tikzstyle{vertexBIG}=[ellipse, draw, scale=0.6, inner sep=3.5pt]
\begin{tikzpicture}[scale=0.6]

 \draw (4,0.6) node {$S_1$};
 \draw [rounded corners] (3.3,1.2) rectangle (4.7,5.8);
 
  \draw (6,0.6) node {$S_2$};
 \draw [rounded corners, top color=gray!10, bottom color=gray!10] (5.3,1.2) rectangle (6.7,5.8);

  \draw (8,0.6) node {$S_3$};
 \draw [rounded corners, top color=gray!10, bottom color=gray!10] (7.3,1.2) rectangle (8.7,5.8);
 
  \draw (10,0.6) node {$S_4$};
 \draw [rounded corners] (9.3,1.2) rectangle (10.7,5.8);

 \draw (12,0.6) node {$S_5$};
 \draw [rounded corners, top color=gray!10, bottom color=gray!10] (11.3,1.2) rectangle (12.7,5.8);

 \draw (14,0.6) node {$S_6$};
 \draw [rounded corners] (13.3,1.2) rectangle (14.7,5.8);

 \draw (16,0.6) node {$S_7$};
 \draw [rounded corners] (15.3,1.2) rectangle (16.7,5.8);
 
  \draw (18,0.6) node {$S_8$};
\draw [rounded corners] (17.3,1.2) rectangle (18.7,5.8);
 
   \draw (20,0.6) node {$S_9$};
\draw [rounded corners, top color=gray!10, bottom color=gray!10] (19.3,1.2) rectangle (20.7,5.8);
 
 \draw (22,0.6) node {$S_{10}$};
 \draw [rounded corners] (21.3,1.2) rectangle (22.7,5.8);

\node (s1) at (22.0,5.2) [vertexX] {$s_1$};
\node (t1) at (14.0,5.2) [vertexX] {$t_1$};
 \node (s2) at (18.0,5.2) [vertexX] {$s_2$};
 \node (t2) at (10.0,5.2) [vertexX] {$t_2$};
 \node (s3) at (10.0,4.4) [vertexX] {$s_3$};
\node (t3) at (4.0,4.4) [vertexX] {$t_3$};
\node (u1) at (4.0,2.5) [vertexX] {};
\node (v1) at (8.0,3.0) [vertexX] {};
 \node (u2) at (10.0,2.5) [vertexX] {};
 \node (v2) at (14.0,2.5) [vertexX] {};
 \node (u3) at (18.0,2.5) [vertexX] {};
\node (v3) at (22.0,2.5) [vertexX] {}; 

\draw [->, line width=0.02cm] (s1) to [out=150, in=30] (t1);
\draw [->, line width=0.02cm] (s2) to [out=150, in=30] (t2);
\draw [->, line width=0.02cm] (s3) to [out=160, in=20] (t3);

 \draw [->, line width=0.05cm, color=blue] (u1) to [out=35, in=155] (v1);
 \draw [->, line width=0.05cm, color=blue] (u1) to [out=0, in=180] (u2);
 \draw [->, line width=0.05cm, color=blue] (u2) to [out=0, in=180] (v2);
 \draw [->, line width=0.05cm, color=blue] (u3) to [out=-15, in=195] (v3);
\draw [->, line width=0.05cm, color=blue] (v2) to [out=15, in=165] (v3);
\draw [->, line width=0.05cm, color=blue] (v1) to [out=15, in=150] (v2);
 \draw [->, line width=0.05cm, color=blue] (u1) to [out=-15, in=195] (u3);

\end{tikzpicture}
\end{center}
\caption{A nice decomposition of a strong semicomplete digraph with three backwards arcs (in thin black). The grey sets ($S_2$, $S_3$, $S_5$, $S_9$) are ignored.
The thick blue arcs are regular-bad.}\label{fig:ignored}
\label{fig2}
\end{figure}

\begin{theorem}
Let $D$ be a strong semicomplete digraph of order at least $4$, $(S_1, \dots , S_p)$ a nice decomposition of $D$ and $(s_1t_1, s_2t_2, \dots , s_rt_r)$ the natural ordering of the backward arcs. 
An arc is contained in a spanning eulerian subdigraph of $D$ if and only if it is good.
\end{theorem}
\begin{proof}
Recall that an arc $uv$ is contained in a spanning eulerian subdigraph of $D$ if and only if there is a spanning $(v,u)$-trail in $D\setminus \{uv\}$.

\medskip

Let us first prove that a bad arc is not contained in any spanning eulerian subdigraph.

Assume first that $uv$ is a regular-bad arc. Let $i_0$ be an integer such that $\ind(u) < i_0 < \ind(v)$ and $S_{i_0}$ is ignored.
Let $j$ be the integer such that $\ind(s_{j+1}) < i_0 < \ind(t_{j-1})$.
Set $L_j = \bigcup_{i=1}^{\ind(s_{j+1})} S_i$, $R_j= \bigcup_{i=\ind(t_{j-1})}^{p} S_i$, and $M_j=V(D) \setminus (L_j\cup R_j)$.
Consider a $(v,u)$-trail $W$ in $D$. It must start in $R_j$ and then use $s_jt_j$, which is the unique arc from $R_j$ to $L_j\cup M_j$. 
But then $W$ cannot return to $R_j\cup M_j$ after using $s_jt_j$, 
as $u \in L_j$ and $s_jt_j$ is the unique arc from $R_j\cup M_j$ to $L_j$.
Hence $W$ is not spanning, because it contains no vertex of $M_j$.
Therefore there is no spanning $(v,u)$-trail in $D\setminus \{uv\}$.

Assume now that $uv$ is a left-bad arc. Since $D$ is semicomplete, $ut_r \in A(D)$. 
By Proposition~\ref{prop:nosame} (i),  $u$ is the unique in-neighbour of $t_r$, and $u$ has in-degree $1$ in $D$.  
Thus any spanning eulerian subdigraph $E$ contains $ut_r$. Moreover $u$ has in- an out-degree $1$ in $E$ and so $E$ does not contain $uv$. 
Similarly, if $uv$ is right-bad, we get that there is no spanning eulerian subdigraph containing $uv$ in $D$. 

\bigskip

We shall now prove by induction on $|D|$ that a good arc $uv$ is contained in a spanning eulerian subdigraph.
This is equivalent to proving the existence of a spanning $(v,u)$-trail in $D\setminus \{uv\}$.
If $|D|=4$, the statement can be easily checked. Therefore, we now assume that $|D| > 4$.

For each $1\leq j<r$, let $N_j$ be a $(t_j, s_{j+1})$-path in $D\langle S_{t_j}\rangle$ if $\ind(t_j)=\ind(s_{j+1})$ and let $N_j=(t_j, s_{j+1})$ otherwise (that is if $\ind(t_j) < \ind(s_{j+1})$).
Let $N=N_1 \cup \cdots \cup N_{r-1}$. Note that $N$ is an $(s_1, t_r)$-path containing all backward arcs. 

\medskip

We first consider the backward arcs.
%Assume first that there is a unique backward arc $uv$. Let $D'=(D\setminus uv) \cup vu$. The digraph $D'$ is not strong and $v$ is an out-generator of $D'$ and $u$ is an in-generator of $D'$. Therefore, by Lemma~\ref{lem:path-reduc}, there is a hamitonian $(v,u)$-path in $D'$ whose union with $uv$ is a hamiltonian cycle in $D$.
%Assume now that there are more than one backward arc ($r\geq 2$).
Let $P_1$ be a hamiltonian path of $D\langle S_1\rangle$ with initial vertex $t_r$ and let $x$ be its terminal vertex. Let $P_p$ be a hamiltonian path of $D\langle S_p\rangle$ with terminal vertex $s_1$ and let $y$ be its initial vertex. Then $Q_1=P_p\cup{}N\cup{}P_1$ is a $(y,x)$-path. 
Observe that in the semicomplete digraph $D - V(Q_1(y,x))$, $x$ has in-degree zero and $y$ has out-degree zero.
Hence, by Lemma~\ref{lem:path-reduc}, there is a hamiltonian $(x,y)$-path $Q_2$ in $D - V(Q_1(y,x))$. Thus $Q_1\cup Q_2$ is a hamiltonian cycle containing all  backward arcs.

\iffalse
Let $P_1$ be a hamiltonian path of $D\langle S_1\rangle$ with initial vertex $t_r$, let $P_p$ be a hamiltonian path of $D\langle S_p\rangle$ with terminal vertex $t_r$, and for $1<i< p$, let $P_i$ be a hamiltonian path of $D\langle S_i\rangle \setminus N$.
Some $P_i$, $1<i< p$, might be empty.
Let $I=\{i \mid V(P_i) \neq \emptyset\}$. For every $i\in I\setminus\{p\}$, let $i^+= \min \{i'\in I \mid i' > i\}$.
Observe that $t(P_i)s(P_{i^+}) \in A(D)$. Hence $N \cup \bigcup_{i\in I} P_i \cup \{t(P_i)s(P_{i^+}) \mid i\in I\setminus \{p\}\}$ is a hamiltonian cycle containing all backward arcs.
\fi

\medskip

Assume now that $uv$ is a flat arc.
In $D$, there are two arc-disjoint $(v,u)$-paths. Indeed, suppose not. By Menger's Theorem, there would be a cut-arc separating $v$ from $u$.
But this cut-arc must be in $D\langle S_u\rangle = D\langle S_v\rangle$, which is strong, contradicting that we have a nice decomposition.
Therefore, by Theorem~\ref{thm:spanning-trail}, there is a spanning $(v,u)$-trail in $D\setminus \{uv\}$.

\medskip

Assume finally that $uv$ is a good forward arc.

\begin{claim}\label{claim:milieu}
If $\ind(u)\geq 3$ or $\ind(u) = 2$ and $|S_1|>1$, then $D$ has a spanning eulerian subdigraph containing $uv$.
\end{claim}
\begin{subproof}
Let $L = \{x \mid \ind(x) < \ind(u)\}$, and $R= \{x \mid \ind(x) \geq \ind(u)\}$, and let $j$ be the integer such that
$s_{j} \in R$ and $t_{j}\in L$.
Let $D_L$ be the digraph obtained from $D\langle L\rangle$ by adding a vertex $z_L$ and all arcs from $L$ to $z_L$ and $z_Lt_{j}$.
Let $D_R$ be the digraph obtained from $D\langle R \rangle$ by adding a vertex $z_R$ and all arcs from $z_R$ to $R$ and $s_{j}z_R$.
Observe that $D_L$ and $D_R$ are strong. Moreover, $(\{z_R\}, S_{\ind(u)}, \dots , S_p)$ is a nice decomposition of $D_R$. Thus $uv$ is neither regular-bad nor a right-bad in $D_R$ for otherwise it would already be regular-bad or right-bad in $D$, and it is not left-bad in $D_R$ because $z_R$ dominates $u$ in this digraph.

Since $\ind(u)\geq 3$ or  $\ind(u) = 2$ and $|S_1|>1$, then $D_R$ is smaller than $D$. Observe moreover that if  $D_R$ is  isomorphic to $D_3$, then the arc $uv$ is in the spanning eulerian subdigraph  $uvz_Ru$.
Therefore, by the induction hypothesis, or this observation, in $D_R$ there is a spanning eulerian subdigraph $E_R$ containing $uv$. Since $z_R$ has in-degree $1$ in $D_R$, $E_R$ contains the arc $s_{j}z_R$ and an arc $z_Ry_R$ for some $y_R\in R$.
By Camion's Theorem, there is a hamiltonian cycle $C_L$ of $D_L$.
It necessarily contains the arc $z_Lt_{j}$ because $z_L$ has out-degree $1$ in $D_L$.
Let $y_L$ be the in-neighbour of $z_L$ in $C_L$. Observe that $y_L\neq t_{j}$, because $|V(D_L)|\geq 3$.
Thus $y_L\dom y_R$, and the union of $C_L - z_L$, $y_Ly_R$, $E_R - z_R$ and $s_{j}t_{j}$ is a spanning eulerian subdigraph of $D$ containing $uv$.
\end{subproof}

By Claim~\ref{claim:milieu}, we may assume that $\ind(u)=1$ or $\ind(u)=2$ and $|S_1|=1$ (that is $S_1=\{t_r\}$).
Similarly, we can assume $\ind(v)=p$ or $\ind(v)=p-1$ and $|S_p|=1$ (that is $S_p=\{s_1\}$).

\begin{claim}\label{claim:ind2}
If $\ind(u) = 2$ and $|S_1|=1$, then $D$ has a spanning eulerian subdigraph containing $uv$.
\end{claim}
\begin{subproof}
Assume first that $r=1$ or $\ind(t_{r-1})>2$.
Let $D_1$ be the strong semicomplete digraph obtained form $D$ by removing $t_r$ and adding the arc $s_ru$. Then $|D_1| = |D|-1 \geq 4$ and $(S_2, \dots , S_p)$ is a nice decomposition of $D_1$.
Consequently, $uv$ is not bad and so, by the induction hypothesis, there is a spanning eulerian subdigraph $W_1$ of $D_1$ containing $uv$. Necessarily, $W_1$ contains $s_ru$ which is a cut-arc in $D_1$. Hence $(W_1\setminus \{s_ru\}) \cup s_rt_ru$ is a spanning eulerian subdigraph of $D$ containing $uv$.

Assume now that $r\geq 2$ and $\ind(t_{r-1})=2$.
Then $D_2=D - t_r$ is strong and $uv$ is good in $D_2$.
Therefore, by the induction hypothesis, there is a spanning eulerian subdigraph $W_2$ of $D_2$ containing $uv$.
If $s_r$ is the tail of an arc $s_rw \in A(W_2 )\setminus \{uv\} $, then $(W\setminus \{s_rw\}) \cup s_rt_rw$ is a spanning eulerian subdigraph of $D$ containing $uv$. 
If not, then $s_r=u$ and $v$ is the only out-neighbour of $u$ on $W_2$. 
Thus $u$ has a unique in-neighbour $z$ in $W_2$.  Since $uv$ is not left-bad, we have $d^-_D(u)\geq 2$.
 Thus $u$ has an in-neighbour $y$ distinct from $z$. If $y=t_r$ then $W_2\cup ut_ru$ is a spanning eulerian subdigraph of $D$ containing $uv$, and if $y\neq t_r$ then $W_2\cup ut_ryu$ is a spanning eulerian subdigraph of $D$ containing $uv$ (Note that $t_ry\in A(D)$ because by Proposition~\ref{prop:nosame} (i), $D$ cannot contain the arc $yt_r$).
\end{subproof}

By Claims~\ref{claim:milieu} and~\ref{claim:ind2}, we may assume that  $\ind(u)=1$ and $\ind(v)=p$.

For every $1\leq i\leq p$, let $P_i$ be a hamiltonian path of $D\langle S_i\rangle$.

Set $t_0=v$ and $s_{r+1}=u$.
For $0\leq j \leq r$, let $X_j=\{x \mid \ind(t_j) \leq \ind(x) \leq \ind(s_{j+1})\}$.
Observe that the $X_j$, $0\leq j \leq r$, form a partition of $V(D)$ because each backward arc is a cut-arc.

\begin{claim}\label{claim:Xj}
For every $0\leq j \leq r$, there is a spanning $(t_j, s_{j+1})$-trail $T_j$ in $D\langle X_j\rangle$.
\end{claim}
\begin{subproof}
Set $i_1=\ind(t_j)$ and $i_2=\ind(s_{j+1})$.

If $i_1 < i_2$, then pick a vertex $x_i$ in each set $S_i$ for $i_1< i < i_2$.
Then $t_jx_{i_1+1} \cdots x_{i_2-1}s_{j+1} \cup \bigcup_{i=i_1}^{i_2} P_i$ is a spanning $(t_j, s_{j+1})$-trail in $D\langle X_j\rangle$.

Assume now that $i_1=i_2$.
There must be two arc-disjoint $(t_j, s_{j+1})$-paths in $D\langle X_j\rangle$, for otherwise, by Menger's Theorem, there is a partition $(T, S)$ of $S_{i_1}$ with $t_j\in T$, $s_{j+1}\in S$ such that there is a unique arc $a$ with tail in $T$ and head in $S$.
But then $a$ would also be a cut-arc of $D$, which is impossible because it is a flat arc.
Now, by Theorem~\ref{thm:spanning-trail}, there is a spanning $(t_j, s_{j+1})$-trail in $D\langle X_j\rangle=S_{i-1}$.
\end{subproof}

Now $\bigcup^r_{i=0} T_j  \cup \{s_jt_j \mid 1\leq j \leq r\} \cup \{uv\}$ is a spanning eulerian subdigraph of $D$ containing $uv$.
\end{proof}

\section{Eulerian spanning subdigraphs avoiding prescribed arcs}\label{sec:avoid}
%%%%%%%%%%%%%%%%%%%%%%%%%%%%%%%%

In this section, we give some support for Conjecture~\ref{Euleravoidkarcs}.

We first prove the existence of a minimum function $f(k)$ such that every $f(k)$-arc-strong semicomplete digraph contains a spanning eulerian subdigraph avoiding any prescribed set of $k$ arcs. 
Conjecture~\ref{Euleravoidkarcs} states that $f(k)=k+1$. 

We need the following theorem. A digraph is {\bf semicomplete multipartite} if it can be obtained from a complete multipartite graph $G=(V,E)$ by replacing each edge $uv\in E$ by either a 2-cycle on $u,v$ or one of the two arcs $uv,vu$.

  \begin{theorem}[Bang-Jensen and Maddaloni~\cite{bangJGT79}]\label{thm:SMDsupereuler}
    A strong semicomplete multipartite digraph has a spanning eulerian subdigraph if and only if it is strong and has an eulerian factor. Furthermore, there exists a polynomial-time algorithm for finding a spanning eulerian subdigraph in a strong semicomplete multipartite digraph $D$ or concluding that $D$ has no eulerian factor.
  \end{theorem}

\begin{theorem}
  \label{thm:f(k)}
  Every semicomplete digraph $D=(V,A)$ with $\lambda{}(D)\geq \frac{(k+1)^2}{4}+1$ has a spanning eulerian subdigraph which avoids any prescribed set of $k$ arcs.
\end{theorem}

\begin{proof} 
Consider a set $A'$ of $k$ arcs and let $X_1,X_2,\ldots{},X_r$, $r\leq k$, be the connected components of $\induce{D}{A'}$. Let $D^*$ be the semicomplete multipartite digraph that we obtain by deleting all arcs of $A$ which lie inside some component $X_i$. Clearly $\alpha{}(D^*)\leq k+1$, and it is easy to see that we did not delete more than $\frac{(k+1)^2}{4}$ arcs across any cut of $D$ so $D^*$ is strong and the claim follows from Theorem~\ref{thm:SMDsupereuler}. 
\end{proof}

Below we shall verify Conjecture \ref{Euleravoidkarcs} for the cases $k=1$, $k=2$ and $k=3$.
\begin{theorem}
  \label{thm:k<=3}
  Let $k\in \{1,2,3\}$.
Every $(k+1)$-arc-strong semicomplete digraph  has a spanning eulerian digraph which avoids any prescribed set of $k$ arcs.
\end{theorem}

\iffalse
For $k=1$ and $2$, it follows directly for the proof of Conjecture \ref{Euleravoidkarcs} in the case when all vertices except possible one are adjacent to at most one prescribed arc (Theorem~\ref{thm:quasi-matching}). 
The case $k=3$ is proved in Subsection~\ref{subsec:k=3}.\fi

We need a number of preliminary results. 

\subsection{Preliminaries}
%%%%%%%%%%%%%%

In this subsection we establish some results for general digraphs that are of independent interest and will be useful in our proofs in the next subsections.

  \subsubsection{Eulerian factors in semicomplete digraphs}\label{subsub:factor}
%%%%%%%%%%%%%%%%

An {\bf eulerian factor} of a digraph $D=(V,A)$ is a spanning subdigraph $H=(V,A')$ so that $d_H^+(v)=d_H^-(v)>0$ for all $v\in V$. 
By a {\bf component} of the eulerian factor $H$  we mean a connected component of the digraph $H$.
Let $\dXY{X}{Y}$ denote the number of arcs from $X$ to $Y$.

\begin{theorem}\label{thm:EfactoravoidX}
A digraph $D$ has no eulerian factor if and only if $V(D)$ can be partitioned into $R_1$, $R_2$ and $Y$ so that
the following hold.

\begin{itemize}
\item $Y$ is independent. 
\item $\dXY{R_2}{Y}=0$, $\dXY{Y}{R_1}=0$ and $\dXY{R_2}{R_1}<|Y|$.
\end{itemize}
\end{theorem}

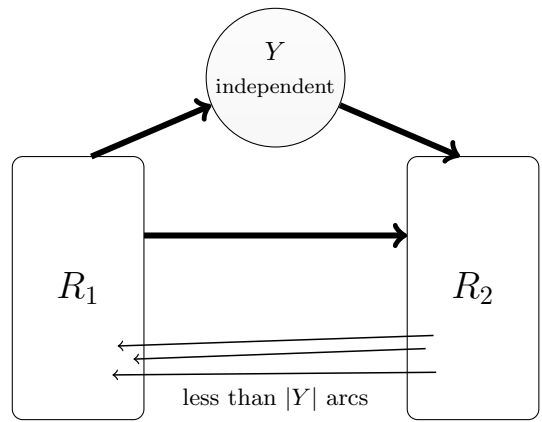
\begin{figure}[hbtp]
\begin{center}
 \tikzstyle{vertexY}=[circle,draw, top color=gray!3, bottom color=gray!6, minimum size=75pt, scale=0.7, inner sep=1.5pt]
 \tikzstyle{vertexBIG}=[ellipse, draw, scale=0.6, inner sep=3.5pt]
\begin{tikzpicture}[scale=0.35]
  \node (y) at (11,14) [vertexY] {};
  \draw [rounded corners] (1,1) rectangle (6,11);
  \draw [rounded corners] (16,1) rectangle (21,11);

 \draw (3.5,6.0) node {{\Large $R_1$}};
 \draw (18.5,6.0) node {{\Large $R_2$}};
 \draw (11,15.0) node {$Y$};
 \draw (11,13.6) node {{\footnotesize independent}};

 \draw [->, line width=0.08cm] (4,11) -- (y);
 \draw [->, line width=0.08cm] (y) -- (18,11);
 \draw [->, line width=0.08cm] (6,8) -- (16,8);
 \draw [->, line width=0.02cm] (17,4.2) -- (5,3.8);
 \draw [->, line width=0.02cm] (16.7,3.7) -- (5.6,3.3);
 \draw [->, line width=0.02cm] (17.1,2.8) -- (4.8,2.7);

 \draw (11,1.8) node {{\small less than $|Y|$ arcs}};

\end{tikzpicture}
\end{center}
\caption{An illustration of Theorem~\ref{thm:EfactoravoidX}. There are no arcs from $R_2$ to $Y$ and no arcs from $Y$ to $R_1$ and less than 
$|Y|$ arcs from $R_2$ to $R_1$.}
\label{figNoFactor}
\end{figure}

\begin{proof}

Let $D=(V,A)$ be any digraph and let 
$B$ be the bipartite digraph obtained from $D$ by splitting every vertex $v$ into an in-going part 
$v^-$ and and out-going part $v^+$. Formally,
$V(B) = \bigcup_{v\in V(D)} \{v^-, v^+\}$ and
$A(B) = \{v^-v^+ \mid v\in V(D)\}\cup \{x^+y^- \mid xy\in A(D)\}$.

Consider the flow network ${\cal N}=(B,l,u)$ with $l$, $u$, being lower and upper bounds on arcs, respectively, such that
$$l(v^-v^+) = 1,  ~~~~ u(v^-v^+) = +\infty ~~~~ l(x^+y^-) = 0,  ~~~~ u(x^+y^-) = 1$$
for every $v\in V(D)$, $xy\in A(D)$.

It is easy to check that there is a one-to-one correspondence between feasible integer-valued circulations on ${\cal N}$ and eulerian factors of $D$.

By Hoffman's circulation theorem~\cite{hoffman1960} (see also Theorem~4.8.2 in \cite{bang2009}), there exists a feasible integer circulation of ${\cal N}$ if (and only if)
\begin{equation}\label{eq:hoffman}
u(\bar{S} , S) \geq  l(S,\bar{S}) 
\end{equation}
for every $S\subseteq V(B)$.

%ssssssss

First assume that $D$ has no eulerian factor, which implies that $u(\bar{S} , S) < l(S,\bar{S})$ for some $S \subseteq V(B)$.
Consider the following possiblities for every $x \in V(D)$ and construct $Y'$, $R_1'$ and $R_2'$ as illustrated below.

\begin{itemize}
\item If $x^- \in S$ and $x^+ \in \bar{S}$, then the arc $x^-x^+$ adds $1$ to $l(S,\bar{S})$, as $l(x^-x^+)=1$. Add $x$ to $Y'$.
\item If $x^+ \in S$ and $x^- \in \bar{S}$, then $u(x^-x^+)=+\infty$ which contradicts $u(\bar{S} , S) < l(S,\bar{S})$. So this case 
cannot happen.
\item If $x^- \in S$ and $x^+ \in S$, then add $x$ to $R_1'$.
\item If $x^- \in \bar{S}$ and $x^+ \in \bar{S}$, then add $x$ to $R_2'$.
\end{itemize}

Note that $R_1'$, $R_2'$ and $Y'$ is a partition of $V(D)$. We note that $l(S,\bar{S}) = |Y'|$, as the lower bound on all arcs except
the $x^-x^+$, $x\in V(D)$, is $0$. We now prove that the following holds.

\begin{equation}\label{eq:u}
u(\bar{S} , S) = \dXY{R_2'}{R_1'} + \dXY{R_2'}{Y'} + \dXY{Y'}{R_1'} + \dXY{Y'}{Y'}
\end{equation}

If $xy$ is an arc from $R_2'$ to $R_1'$ then we note that $u(x^+y^-)=1$ and therefore the arc $xy$ contributes $1$ to $u(\bar{S} , S)$.
Analogously, if $xy$ is an arc from $R_2'$ to $Y'$ or an arc from $Y'$ to $R_1'$, then it also contributes $1$ to $u(\bar{S} , S)$. If
$y_1 y_2$ is an arc within $Y$, then  $u(y_1^+y_2^-)=1$ and therefore the arc $y_1y_2$ also contributes $1$ to $u(\bar{S} , S)$.
As all arcs $x^+ y^-$ in the cut from $\bar{S}$ to $S$ have been counted, this proves Eq.~\eqref{eq:u}.

Now as $u(\bar{S} , S) < l(S,\bar{S}) = |Y'|$), we have

\begin{equation}\label{eq:Y'}
\dXY{R_2'}{R_1'} + \dXY{R_2'}{Y'} + \dXY{Y'}{R_1'} + \dXY{Y'}{Y'} <  |Y'|
\end{equation}

Assume that $Y'$ has minimum size such that Eq.~\eqref{eq:Y'} holds. We will first show that $\dXY{R_2'}{y}=0$ and $\dXY{Y'}{y}=0$ for all
$y \in Y'$. 
Assume for the sake of contradiction that this is not the case and let let $Y^*=Y' \setminus \{y\}$ and
let $R_2^* = R_2' \cup \{y\}$ and let $R_1^*=R_1'$. Then $|Y^*|=|Y'|-1$ and the following holds.

\begin{itemize}
\item $\dXY{R_2^*}{R_1^*} = \dXY{R_2'}{R_1'} + \dXY{y}{R_1'}$.
\item $\dXY{R_2^*}{Y^*} = \dXY{R_2'}{Y'} + \dXY{y}{Y'} - \dXY{R_2'}{y}$.
\item $\dXY{Y^*}{R_1^*} = \dXY{Y'}{R_1'} - \dXY{y}{R_1'}$.
\item $\dXY{Y^*}{Y^*} = \dXY{Y'}{Y'} - \dXY{Y'}{y} - \dXY{y}{Y'}$.
\end{itemize}

Summing up the four above equations we obtain the following (as we assumed that $\dXY{R_2'}{y}\not=0$ or $\dXY{Y'}{y}\not=0$).

\[
\begin{array}{l}
\dXY{R_2^*}{R_1^*} + \dXY{R_2^*}{Y^*} + \dXY{Y^*}{R_1^*} + \dXY{Y^*}{Y^*} \\
\ssX{} =  \dXY{R_2'}{R_1'}  + \dXY{R_2'}{Y'}   + \dXY{Y'}{R_1'}  + \dXY{Y'}{Y'} - \dXY{R_2'}{y}  - \dXY{Y'}{y}  \\
\ssX{} \leq \dXY{R_2'}{R_1'}  + \dXY{R_2'}{Y'}   + \dXY{Y'}{R_1'}  + \dXY{Y'}{Y'} - 1 \\
\ssX{} < |Y'|-1 \\
\ssX{} = |Y^*|  \\
\end{array}
\]

So we note that the partition $(Y^*,R_1^*,R_2^*)$ is a contradiction to the minimality of $Y'$ and we must have $\dXY{R_2'}{y}=0$ and $\dXY{Y'}{y}=0$ 
for all $y \in Y'$. Therefore $\dXY{R_2'}{Y'}=0$ and $\dXY{Y'}{Y'}=0$.
Analogously if $\dXY{y}{R_1'}\not=0$ for some $y \in Y$, then we can let $R_1'' = R_1' \cup \{y\}$, 
$Y'' = Y' \setminus \{y\}$ and $R_2'' = R_2'$ and obtain
the following (as $\dXY{Y'}{Y'}=0$).

\[
\begin{array}{l}
\dXY{R_2''}{R_1''} + \dXY{R_2''}{Y''} + \dXY{Y''}{R_1''} + \dXY{Y''}{Y''} \\
\ssX{} =  \dXY{R_2'}{R_1'}  + \dXY{R_2'}{Y'}   + \dXY{Y'}{R_1'}  + \dXY{Y'}{Y'} - \dXY{y}{R'} \\
\ssX{} \leq \dXY{R_2'}{R_1'}  + \dXY{R_2'}{Y'}   + \dXY{Y'}{R_1'}  + \dXY{Y'}{Y'} - 1 \\
\ssX{} < |Y'|-1 \\
\ssX{} = |Y''|  \\
\end{array}
\]

Therefore $\dXY{Y'}{R_1'}=0$, which implies that  $\dXY{R_2'}{R_1'} <  |Y'|$ and $\dXY{Y'}{R_1'}=\dXY{R_2'}{Y'}=\dXY{Y'}{Y'}=0$.
Therefore we have obtained the desired partition of $V(D)$.

\2

This proved one direction of the theorem.  Now assume that  we can partition the vertices of $V(D)$ 
into $R_1$, $R_2$ and $Y$ such that $Y$ is independent and $\dXY{R_2}{Y}=0$, $\dXY{Y}{R_1}=0$ and $\dXY{R_2}{R_1}<|Y|$.
In this case we note that to get from one vertex of $Y$ to another vertex of $Y$ (or the same vertex of $Y$ with a path of length at least $1$)
we need to use at least one arc from
$R_2$ to $R_1$. However, as $\dXY{R_2}{R_1}<|Y|$, this implies that $D$ cannot contain a eulerian factor (which would contain
at least $|Y|$ arc-disjoint paths between vertices in $Y$).
\end{proof}

\begin{lemma}\label{lem:Efactoravoid}
Let  $k$ a non-negative integer and $D$  a $(k+1)$-arc-strong semicomplete digraph.
Then $D$ has an eulerian factor avoiding any prescribed set of $k$ arcs.
\end{lemma}

\begin{proof}
  Let $A'$ be any set of $k$ arcs in a $(k+1)$-arc-strong semicomplete digraph $D$.
Let $D'=D\setminus A'$ and note that $D'$ is strong.
For the sake of contradiction, assume that $D'$ can be partitioned into $R_1$, $R_2$ and $Y$ such that $Y$ is independent and
$\dXY{R_2}{Y}=0$ and $\dXY{Y}{R_1}=0$ and $\dXY{R_2}{R_1}<|Y|$.  As $D'$ is strong, we must have $R_1 \not= \emptyset$
and $R_2 \not= \emptyset$ and $\dXY{R_2}{R_1} \geq 1$. Therefore $|Y| \geq 2$. Note that at least ${|Y| \choose 2} = |Y|(|Y|-1)/2$
arcs from $A'$ lie completely within $Y$ (as $Y$ is independent in $D'$). Furthermore at least $k+1-(|Y|-1)$ arcs from $A'$ 
go from $R_2$ to $R_1$
as $R_1$ has at least $k+1$ arcs into it in $D$ (and $R_2$ has at least $k+1$ arcs out of it in $D$),
as in $D'$ we have $\dXY{R_2}{R_1} \leq |Y|-1$. So the following holds.

\[
|A'| \geq \frac{|Y|(|Y|-1)}{2} + k - |Y| + 2 = k + 2 + |Y| \left( \frac{|Y| - 3}{2} \right)
\]

The above implies that $|A'| \geq k+1$ (which can easily be verified when $|Y|=2$ and $|Y| \geq 3$), a 
contradiction.  Therefore the partition $(Y,R_1,R_2)$ does not exist and $D'$ has an eulerian factor
by Theorem~\ref{thm:EfactoravoidX}.
\end{proof}

\subsubsection{Merging eulerian subdigraphs}
%%%%%%%%%%%%%%%%%%%%%%%

Let $D=(V,A)$ be a digraph and $D'$ an eulerian subdigraph of $D$ which is not spanning. 
A vertex $x\in V\setminus V(D')$ is {\bf universal to $D'$} 
(or just {\bf universal} when $D'$ is clear from the context) 
if $x$ is adjacent to every vertex of $D'$ and it is {\bf hypouniversal to $D'$} if it is adjacent to all vertices of $D'$ but at most one.
If $x$ has an arc to $D'$ and an arc from $D'$ then we say that $x$ is {\bf mixed to $D'$}.

Let $H$ be a eulerian factor of a digraph $D$ and let $H_1$ and $H_2$ be two distinct components of $H$.
If there exists a spanning eulerian subdigraph, $H^*$ of $\induce{D}{V(H_1) \cup V(H_2)}$, then we say that 
$H_1$ and $H_2$ can be {\bf merged}, as in $H$ we can substitute $H_1$ and $H_2$ by $H^*$ in order to get a 
eulerian factor of $D$ with fewer components.

\begin{lemma}
  \label{lem:merge}
  Let $H_1$ and $H_2$ be two components in an eulerian factor of a digraph $D$ that cannot be merged.  Then all of the following 
points hold for all $i \in \{1,2\}$ and $j=3-i$.
\begin{description}
 \item[(a):] There is no $2$-cycle, $uvu$, where $u \in H_1$ and $v \in H_2$.
 \item[(b):] For every arc $uv \in A(H_i)$ and every $x \in V(H_j)$ we cannot have $ux,xv \in A(D)$.
 \item[(c):] For every arc $uv \in A(H_i)$ and every arc $xy \in A(H_j)$ we cannot have $uy,xv \in A(D)$.
 \item[(d):] If $x \in V(H_i)$ is universal to $H_j$, then $x$ is not mixed to $H_j$.

That is, $x \sdom V(H_j)$ or $V(H_j) \sdom x$.

 \item[(e):] If $x \in V(H_i)$ is hypouniversal and mixed to $H_j$, then there exists a unique $y \in V(H_j)$ such that
$x$ and $y$ are not adjacent and $y^- x, x y^+ \in A(D)$.
\end{description}
\end{lemma}

\begin{proof}
Let $D$, $H_1$, $H_2$ and $i,j$ be defined as in the statement of the lemma.  If there was a $2$-cycle, $uvu$, where $u \in H_1$ and $v \in H_2$, then adding this to $H_1$ and $H_2$ shows that $H_1$ and $H_2$ can be merged, a contradiction. This proves (a).

For the sake of contradiction assume that $uv \in A(H_i)$ and $x \in V(H_j)$ and $ux,xv \in A(D)$. Adding the
arcs $ux$ and $xv$ and removing the arc $uv$ from $H_1 \cup H_2$ shows that $H_1$ and $H_2$ can be merged, a contradiction. This proves (b).

For the sake of contradiction assume that $uv \in A(H_i)$ and $xy \in A(H_j)$ and $uy,xv \in A(D)$. Adding the
arcs $uy$ and $xv$ and removing the arcs $uv$ and $xy$ from $H_1 \cup H_2$ shows that $H_1$ and $H_2$ can be merged, 
a contradiction. This proves (c).

Let $x \in V(H_i)$ be universal to $H_j$ and for the sake of contradiction assume that $x$ is mixed to $H_j$. 
Let the eulerian tour in $H_2$ be $w_1 w_2 w_3 \cdots w_l w_1$ (every arc of $H_2$ is used exactly once).
Without loss of generality we may assume $w_1 x\in A(D)$ (as $x$ is mixed to $H_2$). 
Part (b) implies that $xw_2$ is not an arc in $D$, so $w_2 x \in A(D)$ (as $x$ is universal).
Analogously $w_3 x \in A(D)$. Continueing this process we note that $V(H_2) \dom x$. As there is an arc from 
$x$ to $H_2$ in $D$ (as $x$ is mixed) we have a $2$-cycle between $H_1$ and $H_2$, a contradiction to (a).
This proves (d).

We will now prove (e). Let $x \in V(H_i)$ be hypouniversal and mixed to $H_j$. By (d), vertex  $x$ is not universal 
to $H_j$, so there exists a unique $y \in V(H_j)$ such that $x$ and $y$ are not adjacent. As $x$ is mixed to $H_j$ there
is an arc from $x$ to $V(H_j)$. As $xy \not\in A(D)$, we can assume that $w \in V(H_j)$ 
is chosen such that $xw \in A(D)$ and $xw^- \not\in A(D)$.  By (b) we note that $x$ and $w^-$ are non-adjacent
and therefore $w^-=y$. This implies that $x y^+ \in A(D)$.
Analogously, using (b), we can prove that $y^- x \in A(D)$.
\end{proof}

\subsection{Avoiding a collection of stars}
%%%%%%%%%%%%%%%%%%%%%%%%%%%%%%%%%%%%%%%

If $D$ is a digraph and $A' \subset A(D)$ such that the underlying graph of the digraph induced by $A'$ is a collection of 
stars, then $A'$ is called a 
{\bf star-set} in $D$.
Note that a matching in $D$ is also a star-set.

\begin{lemma}\label{lem:merge-star}
Let $D$ be a semicomplete digraph and let $A' \subset A(D)$ be a star-set in $D$ and let $D'=D\setminus A'$.
If $D'$ is strongly connected  and contains
an eulerian factor with two components $H_1$ and $H_2$ but no spanning eulerian subdigraph, then
the following holds for some $i\in\{1,2\}$ and $j=3-i$.

\begin{description}
 \item[(i):] The eulerian tour in $H_i$ can be denoted by $w_1w_2w_3 \cdots w_l w_1$, such that $w_1$ is not adjacent to 
any vertex in $H_j$ in $D'$.
\item[(ii):] There exists a $k$, such that $R_1 = \{w_2,w_3,\ldots,w_k\}$ and $R_2=\{w_{k+1},w_{k+2},\ldots,w_l\}$ are both
non-empty and the only arc in $D'$ from $R_1$ to $R_2$ is $w_kw_{k+1}$.
\item[(iii):] There is no arc from $R_1$ to $w_1$ and there is no arc from $w_1$ to $R_2$ in $D'$.
\item[(iv):] $V(H_j) \sdom R_1$ and $R_2 \sdom V(H_j)$ in $D'$.
\end{description}
\end{lemma}

\begin{proof}
Let $D$, $A'$, $D'$, $H_1$ and $H_2$ be defined as in the lemma.
Assume that $D'$ has no spanning eulerian subdigraph.
We now prove the following claims.

\begin{claim}\label{claimA}
There must be a vertex in $H_1$ which is not mixed to $H_2$ or a vertex in $H_2$ which is not mixed to $H_1$.
\end{claim}

\begin{subproof}
%{\bf Proof of Claim A:}
Suppose there is no such vertex. Then $D$ contains a cycle $C$ whose vertices alternate between $V(H_1)$ and $V(H_2)$ so taking the union of the arcs of $C$ and  those of  $H_1,H_2$ we obtain a spanning eulerian subdigraph of $D$, contradicting the assumption.
\iffalse For the sake of contradiction assume that every vertex in $H_1$ is mixed to $H_2$ and 
every vertex in $H_2$ is mixed to $H_1$. If there exists a vertex in $V(H_1)$ that is not hypouniversal to $H_2$ or
a vertex in $V(H_2)$ that is not hypouniversal to $H_1$, then let $x$ be such a vertex and
without loss of generality assume that $x \in V(H_1)$. Otherwise let $x \in V(H_1)$ be arbitrary.  Note that by our choice of
$x$ all vertices in $V(H_2)$ that are non-adjacent to $x$ must be hypouniversal (as $A'$ is a star-set). 

Let $w_1 w_2 w_3 \cdots w_l w_1$ be a eulerian tour of $H_2$ and without loss of generality assume that $w_1 x \in A(D)$
(which exists as $x$ is mixed). Let $x w_k \in A(D)$, where $k$ is minimum possible. By Lemma~\ref{lem:merge} (b), and the
minimality of $k$ we note that $x$ and $w_{k-1}$ are non-adjacent.  By the above, $w_{k-1}$ is hypouniversal to $H_1$, which
by Lemma~\ref{lem:merge} (e) implies that $w_{k-1} x^+ \in A(D')$. However the arcs $x w_k$ and $w_{k-1}x^+$ (and $x x^+$ and
$w_{k-1} w_k$) now contradicts Lemma~\ref{lem:merge} (c).\fi
This completes the proof of Claim~\ref{claimA}. 
\end{subproof}

{\bf Definition of $x$}: By Claim~\ref{claimA} we may assume without loss of generality 
that there is a vertex $x \in V(H_1)$ which is not mixed to $H_2$.
Also without loss of generality we may assume that there is no arcs from $H_2$ to $x$.
As $D'$ is strong we can pick $x$ such that $x^+$ has an arc into it from $H_2$ (otherwise consider $x^+$ instead of $x$).

\begin{claim}\label{claimB}
$V(H_2) \sdom x^+$ and $x$ is non-adjacent to every vertex of $H_2$.
\end{claim}

\begin{subproof}
%{\bf Proof of Claim B:}
Let $u x^+$ be an arc from $H_2$ into $x^+$. By Lemma~\ref{lem:merge} (b) and (c), we note that $xu \not\in A(D')$ and
$x u^+ \not\in A(D')$. As there is no arc from $H_2$ to $x$ this implies that $x$ is not adjacent to $u$ or $u^+$.
As $A'$ is a star-set, $x^+$ and $u^+$ are adjacent. By Lemma~\ref{lem:merge} (c), we have $x^+ u^+ \not\in
A(D')$ (as $u x^+ \in A(D')$), which implies that $u^+ x^+ \in A(D')$. We have now shown that $u x^+ \in A(D')$ implies that
$u^+ x^+ \in A(D')$. Analogously we must have $u^{++} x^+ \in A(D')$ and $u^{+++} x^+ \in A(D')$, etc...
Continuing the process we get that  $V(H_2) \dom x^+$.  By Lemma~\ref{lem:merge} (a) there are no $2$-cycles between 
$H_1$ and $H_2$, which implies that $V(H_2) \sdom x^+$. By Lemma~\ref{lem:merge} (b) and the fact that there is no arc from
$H_2$ to $x$ we get that $x$ is non-adjacent to every vertex of $H_2$, completing the proof of Claim~\ref{claimB}.
\end{subproof}

\begin{claim}\label{claimC}
Every vertex in $H_2$ is hypouniversal to $H_1$.  In fact, every vertex in $H_2$ is universal 
to $V(H_1) \setminus \{x\}$.
\end{claim}
\begin{subproof}
%{\bf Proof of Claim~\ref{claimC}:} 
This follows from the fact that $x$ is not adjacent to any vertex in $H_2$ and therefore must be the 
center of a star in $A'$ (as $|V(H_2)| \geq 2$). Therefore all vertices in $H_2$ are leaves in a star in $A'$ and therefore 
have at most one non-neighbour in $D'$. By the above we note that they have exactly one non-neighbour, which in $x$.
\end{subproof}

{\bf Definition}: Let $w_1w_2w_3 \cdots w_lw_1$ be a eulerian tour of $H_1$ and let $w_1=x$.

\begin{claim}\label{claimD}
 The vertex $x$ only appears once in the eulerian tour of $H_1$. That is, in $H_1$ we have
$d^+(x)=d^-(x)=1$.
\end{claim}

\begin{proof}
 Assume for the sake of contradiction that $x$ appears more than once in the 
eulerian tour of $H_1$. As $D'$ is strong there is an arc from $H_1$ to $H_2$, say $w_k u$.
Pick $u$ and $k$ such that $k$ is as large as possible.
As $w_k u\in A(D')$ we note that by Lemma~\ref{lem:merge} (b) $u w_{k+1} \not\in A(D')$.

By the maximality of $k$ this implies that $k=l$ or $u$ and $w_{k+1}$ are non-adjacent.  As  $w_{l+1}=w_1=x$ we note
that in both cases $u$ and $w_{k+1}$ are non-adjacent, which by Claim~\ref{claimC} implies that $w_{k+1}=x$. 
So $u w_{k+2}\in A(D')$ by Claim~\ref{claimB}.
Now deleting the arcs $w_k w_{k+1}$ and $w_{k+1} w_{k+2}$ and adding the arcs $w_k u$ and $u w_{k+2}$ we can merge
$H_1$ and $H_2$ a contradiction.
\end{proof}

{\bf Definition}: As $D'$ is strong there is an arc from $H_1$ to $H_2$, say $w_{k+1} u$.
Pick $u$ and $k$ such that $k$ is as small as possible. Note that $k \geq 2$ as by Claim~\ref{claimB} we have $V(H_2) \sdom w_1$.  
Let $R_1 =  \{w_2,w_3,\ldots,w_k\}$ and let $R_2=\{w_{k+1},w_{k+2},\ldots,w_l\}$.

\begin{claim}\label{claimE}
 $V(H_2) \sdom R_1$ and $R_2 \sdom V(H_2)$ in $D'$.
Note that this proves part (iv) in the lemma.
\end{claim}

\begin{subproof}

By Claim~\ref{claimC} and the minimality of $k$ we note that $V(H_2) \sdom R_1$ holds.

As $w_{k+1} u\in A(D')$, by Lemma~\ref{lem:merge} (b) (and Claim~\ref{claimC}), we have $w_{k+2} u \in A(D')$ or $w_{k+2}=x$.
Continuing this process we note that $R_2 \dom u$. 
By Claim~\ref{claimC}, $u^-$ is universal to $R_2$, so by Lemma~\ref{lem:merge}~(b), we have $R_2 \sdom u^-$.
%We will now show that $R_2 \sdom u^-$. Assume that this is not
%the case and $u^- w_r \in A(D')$ and $w_r \in R_2$. However then $H_1$ and $H_2$ can be merged by Lemma~\ref{lem:merge} (b)
%(considering the arc $u^-u \in H_2$ and the vertex $w_r \in V(H_1)$), a contradiction.  Therefore $R_2 \sdom u^-$.
Analogously $R_2 \sdom u^{--}$. Continuing this process we note that $R_2 \sdom V(H_2)$. 
\end{subproof}

\begin{claim}\label{claimF}
 $R_1 \cap R_2 = \emptyset$ and the only arc from $R_1$ to $R_2$ in $D'$ is $w_kw_{k+1}$.
\end{claim}
Note that this proves part (ii) in the lemma.

\begin{subproof}
If $y \in R_1 \cap R_2$, then by Claim~\ref{claimE} we have $V(H_2) \sdom y$ and $y \sdom V(H_2)$,
which is not possible since $D'$ has no 2-cycle by Lemma~\ref{lem:merge} (a). Therefore $R_1 \cap R_2 = \emptyset$.

Now assume for the sake of contradiction that $uv \in A(D')$ is an arc from $R_1$ to $R_2$ different
from $w_kw_{k+1}$. Note that $uv \not\in A(H_1)$ as $R_1 \cap R_2 = \emptyset$ and all arcs in $H_1$
either lie within $R_1$ or within $R_2$ or are incident with $w_1$ or is the arc $w_k w_{k+1}$. 
Now let $q \in V(H_2)$ be arbitrary and add the arcs $uv, vq, qu$ to $H_1$ and $H_2$ and note that
this merges $H_1$ and $H_2$, a contradiction.
\end{subproof}
\begin{claim}\label{claimG}
There is no arc from $R_1$ to $w_1$ and there is no arc from $w_1$ to $R_2$ in $D'$.
\end{claim}
Note that this proves part (iii) in the lemma.

\begin{subproof}
 For the sake of contradiction assume that $u w_1$ is an arc from $R_1$ to $w_1$.
Let $v \in V(H_2)$ be arbitrary and by Claim~\ref{claimE} note that $vu \in A(D')$. 
We can now merge $H_1$ and $H_2$ by taking the union of the tour $v u w_1 w_2 w_3 \cdots w_l v$ ($w_lv \in A(D')$
by Claim~\ref{claimE}) and $H_2$. This contradiction, implies that there is no arc from $R_1$ to $w_1$ in $D'$.

Analogously we can prove that there is no  arc from $w_1$ to $R_2$ in $D'$.
\end{subproof}

The above claims complete the proof of the lemma, as Claim~\ref{claimB} implies that part (i) of the lemma holds and parts (ii), (iii) and 
(iv) follow from the Claims~\ref{claimE}, \ref{claimF} and \ref{claimG}.
\end{proof}

\begin{theorem}\label{thm:star-set}
Let $D$ be a $(k+1)$-arc-strong semicomplete digraph and let $A' \subset A(D)$ be a star-set of size $k$.
Then $D$ has a spanning eulerian subdigraph which avoids the arcs in $A'$.
\end{theorem}

\begin{proof}
 Let $D'=D\setminus A'$ and note that $D'$ is strong. 
By Lemma~\ref{lem:Efactoravoid}, $D'$ contains an eulerian factor.
Let  $H$ be an eulerian factor of $D'$  with the minimum number of components. Let 
$H_1,H_2, \ldots , H_p$ be the components of $H$, and for every $i\in [p]$ set $D_i=D\langle V(H_i) \rangle$. 
For the sake of contradiction, assume that $p>1$.

Let $T$ be the digraph we obtain from $D'$ by contracting each $V(D_i)$, $i\in [p]$, into one vertex, $x_i$.
Since $D'$ is strong, then $T$ is also strong.
As every $D_i$ contains at least two vertices, $T$ is a semicomplete digraph. 
Let $G$ be the graph with $V(G)=V(T)$ and $uv \in E(G)$ if and only if $uvu$ is a $2$-cycle in $T$.
We now need the following definitions and claims, which completes the proof of the theorem.

\2

{\bf Definition (vital vertex).}  Assume that $x_ix_j$ is an edge in $G$, implying that $x_i x_j x_i$ is a $2$-cycle in $T$.
By the minimality of $p$ the properties of Lemma~\ref{lem:merge-star} hold.
By Lemma~\ref{lem:merge-star}~(i), either there is a vertex in $V(H_i)$ which is not adjacent to any vertex of $H_j$, in which 
case we say that $x_i$ is the vital vertex of the edge $x_ix_j$ in $G$, or
a vertex in $V(H_j)$ which is not adjacent to any vertex of $H_i$,  in which 
case we say that $x_j$ is the vital vertex of the edge $x_ix_j$ in $G$.
Note that $x_i$ and $x_j$ cannot both be vital for $x_ix_j$ as $A'$ is a star-set.
If a vertex is vital for any edge in $G$, then we say that it is a vital vertex in $G$ and otherwise it is non-vital.

\begin{claim}\label{claimA2}
$G$ is a (possibly empty) set of vertex-disjoint stars, where the center vertices of the non-trivial (i.e. of order at least $2$) stars  are 
exactly the vital vertices of $G$.
\end{claim}

\begin{subproof} Assume that $x_i x_j \in E(G)$ and that $x_i$ is the vital vertex of $x_i x_j$. 
That is, there is a $w_1 \in V(H_i)$ which is not adjacent to any vertex of $H_j$.
We will now show that $d_G(x_j)=1$.  That is, $x_ix_j$ is the only edge in $G$ touching $x_j$. Assume for the sake of contradiction
that $x_kx_j$ is an edge in $G$ with $k \not= i$. As $A'$ is a star-set we note that $x_k$ cannot be the vital vertex for
$x_kx_j$ and $x_j$ also cannot be the vital vertex. This implies that $d_G(x_j)=1$. 

So for every edge in $G$ one endpoint is the vital vertex and the other endpoint has degree one. 
This implies that $G$ is a vertex-disjoint collection of stars, where the center vertices of the stars are 
exactly the vital vertices of $G$, which completes the proof of Claim~\ref{claimA2}.
\end{subproof}

\begin{claim}\label{claimB2}
If there exists a $3$-cycle $x_i x_j x_k x_i$ in $T$ such that $x_j x_i \not\in A(T)$ and $x_i x_k \not\in A(T)$ and
there is a vertex $u \in V(H_k)$ that is dominated by all of $V(H_j)$, except for possibly one vertex, then $H_i$, $H_j$ and
$H_k$ can be merged.
\end{claim}

\begin{subproof} For the sake of contradiction assume w.l.o.g. that $i=1$, $j=2$ and $k=3$ in the statement of the claim.
That is, $x_1 x_2 x_3 x_1$ is a $3$-cycle in $T$  and $x_2 x_1 \not\in A(T)$ and $x_1 x_3 \not\in A(T)$ and
there is a vertex $u \in V(H_3)$ that is dominated by all of $V(H_2)$, except for possibly one vertex.
Let $W=H_1 \cup H_2 \cup H_3$.

As $A'$ is a star-set and there is no arc from $H_1$ to $H_3$ we note that either $u$ has an arc out of it to $H_1$ or 
$u^-$ has an arc out of it to $H_1$. Consider the two possibilities below.

\begin{itemize}
\item If there is an arc $uv$ with $v\in V(H_1)$, then add $uv$ to $W$.
\item Otherwise there exists an arc $u^- v \in A(D')$ with $v \in V(H_1)$ and add the arc $u^- v$ to $W$ and delete the arc $u^-u$ from $W$.
\end{itemize}

The new $W$ now has $d^+(a)=d^-(a)$ for all $a \in V(W)\setminus \{u,v\}$ and $d^+(u) = d^-(u)+1$ and $d^-(v)=d^+(v)+1$.
Analogously to above there is an arc from $v$ to $H_2$ or from $v^-$ to $H_2$. Again consider the two possibilities below.

\begin{itemize}
\item If there is an arc $vw$ with $w\in V(H_2)$, then let $vw \in A(D')$ be such an arc and add $vw$ to $W$.
\item Otherwise there exists an arc $v^- w \in A(D')$ with $w \in V(H_2)$ and
add the arc $v^- w$ to $W$ and delete the arc $v^-v$ from $W$.
\end{itemize}

Analogously to above the new $W$ now has $d^+(a)=d^-(a)$ for all $a \in V(W)\setminus \{u,w\}$ 
and $d^+(u) = d^-(u)+1$ and $d^-(w)=d^+(w)+1$.
Note that there is an arc from $w$ to $u$ or from  $w^-$ to $u$, as $u$ was 
dominated by all of $V(H_2)$, except for possibly one vertex.

\begin{itemize}
\item If there is an arc from $w$ to $u$, then add $wu$ to $W$.
\item Otherwise $w^- u \in A(D')$ and 
add the arc $w^- u$ to $W$ and delete the arc $w^-w$ from $W$.
\end{itemize}

Now $W$ is a spanning eulerian subdigraph of $\induce{D}{V(H_1) \cup V(H_2) \cup V(H_3)}$, contradicting the minimality of $p$,
and thereby proving  Claim~\ref{claimB2}.
\end{subproof}

\begin{claim}\label{claimC2}
There is no induced $3$-cycle in $T$.
\end{claim}

\begin{subproof} For the sake of contradiction assume $x_1 x_2 x_3 x_1$ is an induced $3$-cycle in $T$.
That is, $x_1x_3, x_3x_2, x_2 x_1 \not\in A(T)$. If there is no vertex in $H_3$ that is hypouniversal to $H_2$, then
all vertices in $H_2$ are hypouniversal to $H_3$, as $A'$ is a star-set.  So we can assume without loss of generality 
that there is a vertex  $u \in V(H_3)$ that is hypouniversal to $H_2$ (otherwise reverse all arcs and rename $H_1$, $H_2$ and
$H_3$). Claim~\ref{claimB2} now implies that $H_1$, $H_2$ and $H_3$ can be merged, a contradiction.
This proves Claim~\ref{claimC2}.
\end{subproof}

\begin{claim}\label{claimD2}
  There is no vertex in $G$ of degree $p-1$ (that is, $G$ does not consist of one spanning star).
\end{claim}

\begin{subproof} Assume for the sake of contradiction that $x \in V(G)$ has degree $p-1$ in $G$.
By Claim~\ref{claimA2} and Claim~\ref{claimC2} we note that $T-x$ is a transitive tournament, so without loss of generality assume 
that $x=x_1$  and $x_2,x_3,\ldots,x_p$ are named such that if $2 \leq i < j \leq p$ then $x_ix_j \in A(T)$ 
(and $x_j x_i \not\in A(T)$). By Claim~\ref{claimA2} we may assume that $x_1$ is the vital vertex for all edges $x_1 x_i$, 
$i \in \{2,3,\ldots,p\}$ in $G$.

Consider the $2$-cycle $x_1 x_2 x_1 \in T$. By Lemma~\ref{lem:merge-star} the following holds.

\begin{description}
 \item[(i):] The eulerian tour in $H_1$ can be denoted by $w_1w_2w_3 \cdots w_l w_1$, such that $w_1$ is not adjacent to
any vertex in $H_2$ in $D'$.
\item[(ii):] There exists a $k$, such that $R_1 = \{w_2,w_3,\ldots,w_k\}$ and $R_2=\{w_{k+1},w_{k+2},\ldots,w_l\}$ are both
non-empty and the only arc in $D'$ from $R_1$ to $R_2$ is $w_kw_{k+1}$.
\item[(iii):] There is no arc from $R_1$ to $w_1$ and there is no arc from $w_1$ to $R_2$ in $D'$.
\item[(iv):] $V(H_2) \sdom R_1$ and $R_2 \sdom V(H_2)$ in $D'$.
\end{description}

In $H_1$ we note that the only arc into $R_2$ is $w_k w_{k+1}$. We now consider the cases when there 
is an arc into $R_2$ in $D' \setminus w_k w_{k+1}$ and when there is no such arc. 

\2 

{\em Case 1. There is an arc into $R_2$ in $D' \setminus w_k w_{k+1}$.}  In this case assume that $uv$ is such an arc and note that
$u \in H_j$ for some $j \in \{3,4,\ldots,p\}$ and $v \in R_2$. Let $q \in V(H_2)$ be arbitrary and note that $vq \in A(D')$ (as 
$v \in R_2$ and $R_2 \sdom V(H_2)$, by (iv) above) and $qu \in A(D')$ (as $V(H_2) \sdom V(H_j)$, as $A'$ is a star-set).  Therefore,
$uvqu$ is a $3$-cycle in $D'$ and adding this $3$-cycle to $H_1$, $H_2$ and $H_j$ merges them, a contradiction to the 
minimality of $p$.

\2

{\em Case 2. There is no arc into $R_2$ in $D' - w_k w_{k+1}$.}  In this case consider $A''$, which consists of all the arcs in 
$A'$ except the arcs between $w_1$ and $V(H_2)$.  As there are at least two arcs between $w_1$ and $V(H_2)$ in $A'$ (as
$|V(H_2)| \geq 2$) we have $|A''| \leq |A'|-2 \leq (k+1)-2 = k-1$. 
Furthermore $w_k w_{k+1}$ is the only arc into $R_2$ in $D\setminus A''$, which implies that $D$ is at most $k$-arc-connected, a
contradiction.

\end{subproof}

\begin{claim}\label{claimE2}
There exists a $3$-cycle, say $x_1 x_2 x_3 x_1$, in $T$ such that $x_2 x_1 \in A(T)$ and $x_3 x_2 \not\in A(T)$ and
$x_1 x_3 \not\in A(T)$.
\end{claim}

\begin{subproof} If $|V(G)|=2$, then as $D'$ is strong and therefore also $T$, we note that $T$ consists of a $2$-cycle.  
However this is a contradiction to Claim~\ref{claimD2}. Therefore we may assume that $|V(G)|=|V(T)| \geq 3$. As by Claim~\ref{claimC2} $T$ does not contain 
an induced $3$-cycle, $G$ must contain a star $S$, and by Claim~\ref{claimD2} a vertex $x \in V(T) \setminus V(S)$.
Let $y$ be the center of the star $S$ (if $|E(S)|=1$ let $y \in V(S)$ be arbitrary) and without loss of generality assume that 
$yx \in A(T)$. Let $P=p_1p_2\ldots p_l$ be a shortest path from $x$ ($x=p_1$) to $y$ ($p_l=y$) in $T$. (Such a path exists because $T$ is strong.)
By the minimality of $l$ note that $y \sdom \{p_1,p_2,\ldots,p_{l-2}\}$ and as $x \not\in V(S)$ we have $l \geq 3$.

Therefore $c=p_{l-2} p_{l-1} p_l p_{l-2}$ is a $3$-cycle in $T$ and $p_{l-2} \not\in V(S)$ (as $p_{l-2} p_l \not\in A(T)$).
Therefore $p_l p_{l-2}$ is not an edge in $G$. If $p_{l-1} \in V(S)$ then $p_{l-2} p_{l-1} \not\in E(G)$ and if
$p_{l-1} \not\in S$ then $p_{l-1} p_l \not\in E(G)$. So in both cases $C$ has at most one arc belonging to a $2$-cycle.
By Claim~\ref{claimC2} we note that there is exactly one arc belonging to a $3$-cycle, thereby proving Claim~\ref{claimE2}.
\end{subproof}

One can now prove the theorem.
By Claim~\ref{claimE2}, we may let  $x_1 x_2 x_3 x_1$ be a $3$-cycle in $T$ such that $x_2 x_1 \in A(T)$ 
and $x_3 x_2 \not\in A(T)$ and $x_1 x_3 \not\in A(T)$.
By Lemma~\ref{lem:merge-star} either there is a vertex in $H_2$ that is not adjacent to any vertex in $H_1$ or there
is a vertex in $H_1$ that is not adjacent to any vertex in $H_2$. By reversing all arcs if necessary, we may assume that 
$x \in V(H_2)$ is not adjacent to any vertex of $H_1$.  This implies that $x^- \sdom V(H_1)$ and $V(H_1) \sdom x^+$,
by Lemma~\ref{lem:merge-star}. As $x_2 x_3, x_3 x_1 \in A(T)$ and $x_3 x_2, x_1 x_3 \not\in A(T)$ and $V(H_1) \sdom x^+$ 
it follows from Claim~\ref{claimB2}  that $H_1$, $H_2$ and $H_3$ can be merged, a contradiction.

This completes the proof. 
\end{proof}

 Since a set of at most two arcs always form a star-set, we have the following corollary.

\begin{corollary} 
\label{cor:k<=2}
Every $2$-arc-strong semicomplete digraph  has a spanning eulerian digraph which avoids any prescribed arc and every $3$-arc-strong semicomplete digraph  has a spanning eulerian digraph which avoids any set of two prescribed arcs.
\end{corollary}

\subsection{Avoiding three arcs}\label{subsec:k=3}
%%%%%%%%%%%%%%%%%%%%%%%%%%

\begin{theorem} 
\label{thm:k=3}
Every $4$-arc-strong semicomplete digraph  has a spanning eulerian digraph which avoids any set of three prescribed arcs.
\end{theorem}

\begin{proof}

Let $D=(V,A)$ be a 4-arc-strong semicomplete digraph, let $F =\{a,a',a''\}\subset A$ be a set of three arcs and let $D'=D\setminus F$.
By Theorem \ref{thm:star-set} we may assume that 
the graph $N$ induced by non-edges of $D'$ is either a triangle or the path $P_4$ on four vertices. If $N$ is a triangle, then $D'$ is semicomplete multipartite and the claim  follows from Theorem \ref{thm:SMDsupereuler}, so the only remaining case is that $N$ is a $P_4$.\\

 By Lemma~\ref{lem:Efactoravoid}, $D'$ contains an eulerian factor.
Let  ${\cal E}$ be an eulerian factor of $D'$  with the minimum number of components. Let 
$H_1,H_2, \ldots , H_p$ be the components of ${\cal E}$, and for every $i\in [p]$ let $W_i$ be a closed spanning trail of $H_i$.
For the sake of contradiction, assume that $p>1$.
If $D'$ contains a cycle $C$ all of whose arcs go between different components of ${\cal E}$, then by adding the arcs of $C$ we obtain a better eulerian factor, contradicting the choice of ${\cal E}$. Hence we may assume w.l.o.g. that $H_1$ contains a vertex $v$ with no arc into it from any other $H_j$. As $D'$ is strong we can furthermore assume that the successor $v^+$ of $v$ on $W_1$ has an arc into it from another $H_j$ and by renumbering if necessary we can assume that there is a vertex $u$ of $H_2$ such that $uv^+$ is an arc of $D'$. Let $u^+$ be the successor of $u$ on $W_2$. Since $\induce{D'}{V(H_1)\cup{}V(H_2)}$ has no spanning closed trail it follows from Lemma \ref{lem:merge}~(c) that $v$ is non-adjacent to both $u$ and
$u^+$.

If $v^+$ and $u^+$ are adjacent in $D'$, then we must have $u^+v^+\in A(D')$ by Lemma \ref{lem:merge}~(b), and  now since $v$ dominates $V(H_2)-\{u,u^+\}$ we have $V(H_2)=\{u,u^+\}$ for otherwise the arcs $vu^{++},u^+v^+$ contradict Lemma~\ref{lem:merge}~(c). If $v^+$ and $u^+$ are not adjacent in $D'$, then $V(N)=\{u,u^+,v,v^+\}$.

Suppose first that $p>2$. It follows from the minimality of $p$ and the fact that $N$ is a $P_4$ that we must have $V(H_1)\mapsto V(H_3)\cup\ldots\cup{}V(H_p)$. Suppose there is an arc $zw\in A(D')$ from
$V(H_i)$ to $V(H_2)$ for some $i>2$. If $u^+v^+\in A(D')$, then $w\in \{u,u^+\}$ by the argument above and thus
$vzwv^+$ is a path in $D'$ which shows that $H_1,H_2,H_i$ can be replaced by one eulerian subdigraph, contradicting the choice of $\cal E$. So $u^+$ and $v^+$ must be non-adjacent as otherwise there is no arc entering $V(H_2)$, contradicting that $D'$ is strong. As remarked above , this means that  $V(N)=\{u,u^+,v,v^+\}$ and hence every vertex of $V(H_i)$ is adjacent to every vertex of $V(H_2)$ so by
Lemma \ref{lem:merge} and the choice of $\cal E$ we must have $V(H_i)\mapsto V(H_2)$. Now $v^+zuv^+$ is a 3-cycle in $D'$ which shows that we can merge $W_1,W_2,W_i$, contradicting the minimality of $p$. So we must have $p=2$.

Suppose first that $|V(H_2)|>2$. By the remark above, $V(N)=\{u,u^+,v,v^+\}$. Hence $v^+$ is adjacent to all vertices of $R=V(H_2)\setminus \{u,u^+\}$ and since $v$ dominates all of these, we also conclude from Lemma \ref{lem:merge} and the minimality of $p$ that $v^+\mapsto R$ and we see that
$V(H_1)\mapsto R$. Let $u^{++}$ be the successor of $u^+$ on $W_2$. If $H_2$ has a spanning $(u^{++},u)$-trail $T$ then we can insert $V(H_2)$ in $W_1$ by deleting the arc $vv^+$ and adding the arcs of the trail $vu^{++}T[u^{++},u]uv^+$, contradicting the minimality of $\cal E$. Thus there is no spanning $(u^{++},u)$-trail in $H_2$ and, by  Theorem \ref{thm:spanning-trail} and Menger's theorem, we can partition $V(H_2)$ into two sets $Z_1,Z_2$ such that $u^{++}\in Z_1,u\in Z_2$ and there is precisely one arc from $Z_1$ to $Z_2$ in $D_2$. But then there are at most three arcs leaving $Z_1$ in $D$, contradicting that $D$ is 4-arc-strong.

Henceforth $V(H_2)=\{u,u^+\}$. As $D$ is 4-arc-strong this implies that $|V(H_1)|>2$. 
Note that if $u^+v^+\in A(D')$, then we may assume, by renaming $u,u^+$ if necessary, that $u$ is adjacent to all vertices of $V(H_1)\setminus \{v,v^+\}$ in $D'$. This holds automatically if $V(N)=\{u,u^+,v,v^+\}$.

 If $uv^-$ is an arc of $D$ (and hence of $D'$), then it follows from Lemma \ref{lem:merge}~(b)
 that  $u\mapsto V(H_1)\setminus \{v\}$, contradicting that the in-degree of $u$ is at least 4 in $D$. 

Hence
$v^-u\in A(D')$. This implies that  either $u^+$ and $v^-$ are non-adjacent or $v^-\mapsto u^+$ by Lemma~\ref{lem:merge}~(b). As  $D$ is 4-arc-strong the vertex $u$ has at least two in-neighbours and two out-neighbours in $V(H_1)$ in $D'$. This and
the minimality of $p$ implies that there exist a vertex $w\in V(H_1)$ such that $u\mapsto Y$ and $X\mapsto u$, where $Y= V(W_1[v^+,w^-])$ and $X=V(W_1[w,v^-])$. It is easy to see that we also have $u^+\mapsto Y\setminus \{v^+\}$ and if $u^+$ is not adjacent to $v^-$ then $u^+\mapsto Y$. Now we conclude that $H_1$ has no spanning $(v^+,v^-)$-trail $T'$ as otherwise
either $uv^+T'[v^+,v^-]v^-u^+u$ or $u^+v^+T'[v^+,v^-]v^-uu^+$ would be a closed spanning trail of $D$. As $v^-v,vv^+\in A(D')$ and $v$ is adjacent to all vertices of $V(H_1)\setminus v$ and cannot be inserted in the trail $W_1[v^+,v^-]$, there exists a vertex $z\in V(H_1)\setminus v$ such that $v\mapsto W_1[v^+,z^-]$ and $W_1[z,v^-]\mapsto v$. By symmetry we can assume that $z\in X$ and thus $v\mapsto Y$.

As $D$ is 4-arc-strong, by Menger's theorem,
there are at least four arcs with tail in $Y$ and head in $V(D)\setminus Y$.
As we have $\{u,u^+\}\mapsto Y-v^+$, $u\mapsto v^+$ and  $v\sdom Y$ the head of at least three of those arcs must be in $X$.
Consequently, in $H_1$ there are at least three arcs with tail in $Y$ and head in $X$.
In particular, there are $y\in Y, x\in X$ such that $yx$ is not the arc $v^+v^-$ and there are two arc-disjoint $(y,x)$-paths in $H_1$. 
Thus by Theorem~\ref{thm:spanning-trail}, there exists a spanning $(y,x)$-trail $T_1$ in $H_1$.
Now  either $T_1[y,x]xu^+uy$  or $T_1[y,x]xuu^+y$ (or both) is a spanning eulerian trail of $D'$, a contradiction.

\iffalse
If $H_1$ has a spanning $(v^+,v^-)$-trail we obtain a contradiction to the minimality of $\cal E$ so, by Menger's theorem and Theorem \ref{thm:spanning-trail},  there exists a partition $W_1,W_2$ of $V(H_1)$ with $v^+\in W_1,v^-\in W_2$ such that there is precisely one arc from $W_1$ to $W_2$.
\fi

\end{proof}

 \section{Unavoidable arcs in semicomplete digraphs}\label{sec:unavoid}
 %%%%%%%%%%%%%%%%%%%%%%%%%%%

Let $D$ be a strong semicomplete digraph with at least one cut-arc (so $\lambda{}(D)=1$)
 An arc $a$ is {\bf unavoidable} if it is contained in all spanning eulerian subdigraphs of $D$ (so $D\setminus a$ has no spanning closed trail).
 Observe that every cut-arc is unavoidable. 
 %In particular, if $D$  strong of order $2$ or $3$, then all its arcs are unavoidable.
%Below we shall give two different characterizations of the set of unavoidable arcs in a given semicomplete digraph $D$.

The following is a direct consequence of Theorem \ref{thm:SMDsupereuler}. Note that if $D\setminus a$ is semicomplete then it has a hamiltonian cycle and we can find such a cycle in polynomial time in any semicomplete digraph.
\begin{corollary}\label{thm:algo}
There is a polynomial-time algorithm that, given a semicomplete digraph $D$ and an arc $a$, decides whether $a$ is unavoidable in $D$ and returns a spanning eulerian subdigraph avoiding $a$ when one exists.
\end{corollary}

We believe that Corollary \ref{thm:algo} can be generalized to the following.

\begin{conjecture}
  \label{conj:avoidkalg}
  For each fixed positive integer $k$, there exists a polynomial-time algorithm which, given a semicomplete digraph $D=(V,A)$ and $A'\subset A$ with $|A'|=k$, decides whether $D\setminus A'$ has a spanning eulerian subdigraph.
\end{conjecture}

The analogous conjecture for hamiltonian cycles was posed in \cite[Conjecture 7.4.14]{bang2009} and is still open for $k\geq 2$. For $k=1$ a polynomial-time algorithm follows from  \cite{bangJA13}.

\subsection{A classification of the set of unavoidable arcs}
%%%%%%%%%%%%%%%%%%%

In this subsection  we give a complete characterization of the pairs $(D,a)$ such that $D$ is a semicomplete digraph in which $a$ is an unavoidable arc.
We shall the following theorem.

% tttttttttt
\begin{theorem}\label{thm:charac-unavoid}
Let $D$ be a semicomplete digraph and let $uv \in A(D)$ be arbitrary and let $D'=D\setminus \{uv\}$.
If $D'$ is strong, then $D'$ contains a spanning eulerian subdigraph if and only if $V(D')$ cannot be 
partitioned into $R_1$, $R_2$ and $Y=\{u,v\}$ such that $Y$ is independent,
$\dXY{R_2}{Y}=0$, $\dXY{Y}{R_1}=0$ and $\dXY{R_2}{R_1}=1$.

%Note that if $D'$ can be partitioned as above then $Y=\{u,v\}$ as $D'$ is strong so $|Y| \geq 2$. 
\end{theorem}

\begin{proof}
 Let $D'$ be defined as in the theorem.  If $D'$ can be partitioned into $R_1$, $R_2$ and $Y$ such that $Y$ is independent and
$\dXY{R_2}{Y}=0$ and $\dXY{Y}{R_1}=0$ and $\dXY{R_2}{R_1}<|Y|$, then we must have $Y=\{u,v\}$ since we only deleted one arc from a semicomplete digraph and $\dXY{R_2}{R_1}>0$ as $D'$ is strong. Now it follows from Theorem~\ref{thm:EfactoravoidX} that $D'$ contains no eulerian factor
 and therefore also no spanning eulerian subdigraph. So assume that $D'$ cannot be partitioned in this way, which
 by Theorem~\ref{thm:EfactoravoidX} implies that $D'$ contains an eulerian factor. $D'$ is clearly a semicomplete multipartite digraph so it follows from Theorem \ref{thm:SMDsupereuler} that $D'$ has a spanning eulerian subdigraph.

\end{proof}

\2

We first observe that the backward arcs with respect to a nice decomposition are unavoidable since they are cut-arcs.

 \begin{proposition}\label{backward-unavoid}
 Let $D$ be a strong semicomplete digraph  of order at least $4$  and let $(S_1, \dots , S_p)$ be a nice decomposition of $D$.
Every backward arc is unavoidable. 
 \end{proposition}

 If $D$ is a semicomplete digraph with vertex set $\{a,b,c,d\}$ such that $\{ab, bc, cd, ad, ca, db\}  \subseteq A(D) \subseteq \{ab, bc, cd, ad, ca, db, cb\}$, then
 the arc $ad$ is {\bf exceptional}. See Figure~\ref{fig:exceptional}.
  \begin{figure}[hbtp]
\begin{center}
\tikzstyle{vertexX}=[circle,draw, top color=gray!5, bottom color=gray!30, minimum size=16pt, scale=0.6, inner sep=0.5pt]
\tikzstyle{vertexY}=[circle,draw, top color=gray!5, bottom color=gray!30, minimum size=20pt, scale=0.7, inner sep=1.5pt]

%\tikzstyle{vertexBIG}=[ellipse, draw, scale=0.6, inner sep=3.5pt]
\begin{tikzpicture}[scale=0.8]
\node (a1) at (0,0) [vertexX] {a};
\node (b1) at (2,0) [vertexX] {b};
 \node (c1) at (4,0) [vertexX] {c};
 \node (d1) at (6,0) [vertexX] {d};

\draw [->, line width=0.02cm] (a1) to (b1);
\draw [->, line width=0.02cm] (b1) to (c1);
\draw [->, line width=0.02cm] (c1) to (d1);
\draw [->, line width=0.02cm] (c1) to [out=-160, in=-20] (a1);
\draw [->, line width=0.02cm] (d1) to [out=-160, in=-20] (b1);
 \draw [->, line width=0.05cm, color=blue] (a1) to [out=20, in=160] (d1);

\node (a2) at (10,0) [vertexX] {a};
\node (b2) at (12,0) [vertexX] {b};
 \node (c2) at (14,0) [vertexX] {c};
 \node (d2) at (16,0) [vertexX] {d};

\draw [->, line width=0.02cm] (a2) to (b2);
\draw [->, line width=0.02cm] (b2) to (c2);
\draw [->, line width=0.02cm] (c2) to (d2);
\draw [->, line width=0.02cm] (c2) to [out=160, in=20] (b2);
\draw [->, line width=0.02cm] (c2) to [out=-160, in=-20] (a2);
\draw [->, line width=0.02cm] (d2) to [out=-160, in=-20] (b2);
 \draw [->, line width=0.05cm, color=blue] (a2) to [out=20, in=160] (d2); 
 
\end{tikzpicture}
\end{center}

  \caption{The two digraphs having an exceptional arc ($ad$ in thick blue).}\label{fig:exceptional}

  \end{figure}
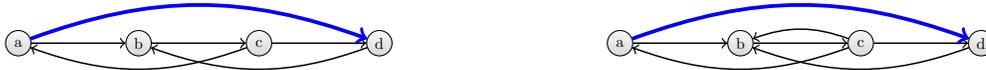

 Let $D$ be a semicomplete digraph  of order at least $4$ and let $(S_1, \dots , S_p)$ be a nice decomposition of $D$.
 A forward arc $uv$ is {\bf regular-compulsory} if there is an index $i$ such that $1<i<p-1$, $S_i=\{u\}$, $S_{i+1}=\{v\}$, and both $S_i$ and $S_{i+1}$ are ignored.
If $|S_i|=1$ for all $1\leq i\leq 3$, say $S_i=\{v_i\}$, and $v_2v_1\in A(D)$, $v_1v_2 \notin A(D)$, $v_3v_2 \notin A(D)$ and $N^-(v_3)=\{v_1, v_2\}$, then the arc $v_1v_3$ is {\bf left-compulsory}.
If $|S_i|=1$ for all $p-2\leq i\leq p$, say $S_i=\{v_i\}$, and $v_pv_{p-1}\in A(D)$, $v_{p-1}v_{p} \notin A(D)$, $v_{p-1}v_{p-2} \notin A(D)$, $N^+(v_{p-2})=\{v_{p-1}, v_p\}$, then $v_{p-2}v_{p}$ is {\bf right-compulsory}. See Figure~\ref{fig:compulsory}.

\begin{figure}[hbtp]
\begin{center}
\tikzstyle{vertexX}=[circle,draw, top color=gray!5, bottom color=gray!30, minimum size=16pt, scale=0.6, inner sep=0.5pt]
\tikzstyle{vertexY}=[circle,draw, top color=gray!5, bottom color=gray!30, minimum size=20pt, scale=0.7, inner sep=1.5pt]

%\tikzstyle{vertexBIG}=[ellipse, draw, scale=0.6, inner sep=3.5pt]
\begin{tikzpicture}[scale=0.6]

 \draw (4,0.6) node {$S_1$};
% \draw [rounded corners] (3.3,1.2) rectangle (4.7,5.8);
 
  \draw (6,0.6) node {$S_2$};
% \draw [rounded corners, top color=gray!10, bottom color=gray!10] (5.3,1.2) rectangle (6.7,4.8);

  \draw (8,0.6) node {$S_3$};
 %\draw [rounded corners, top color=gray!10, bottom color=gray!10] (7.3,1.2) rectangle (8.7,4.8);
 
  \draw (10,0.6) node {$S_4$};
 \draw [rounded corners] (9.3,1.2) rectangle (10.7,4.8);

 \draw (12,0.6) node {$S_5$};
% \draw [rounded corners, top color=gray!10, bottom color=gray!10] (11.3,1.2) rectangle (12.7,5.8);

 \draw (14,0.6) node {$S_6$};
% \draw [rounded corners] (13.3,1.2) rectangle (14.7,5.8);

 \draw (16,0.6) node {$S_7$};
 \draw [rounded corners] (15.3,1.2) rectangle (16.7,4.8);
 
  \draw (18,0.6) node {$S_8$};
\draw [rounded corners] (17.3,1.2) rectangle (18.7,4.8);
 
   \draw (20,0.6) node {$S_9$};
\draw [rounded corners, top color=gray!10, bottom color=gray!10] (19.3,1.2) rectangle (20.7,4.8);
 
 \draw (22,0.6) node {$S_{10}$};
 \draw [rounded corners] (21.3,1.2) rectangle (22.7,4.8);

\node (s1) at (22.0,4.2) [vertexX] {$s_1$};
\node (t1) at (16.0,4.2) [vertexX] {$t_1$};
 \node (s2) at (18.0,4.2) [vertexX] {$s_2$};
 \node (t2) at (10.0,4.2) [vertexX] {$t_2$};
 \node (s3) at (10.0,3.4) [vertexX] {$s_3$};
\node (v1) at (4.0,3) [vertexX] {$v_1$};
\node (v2) at (6.0,3) [vertexX] {$v_2$};
 \node (u) at (12.0,3) [vertexX] {$u$};
 \node (v) at (14.0,3) [vertexX] {$v$};
\node (v3) at (8.0,3) [vertexX] {$v_3$};

\draw [->, line width=0.02cm] (s1) to [out=150, in=30] (t1);
\draw [->, line width=0.02cm] (s2) to [out=150, in=30] (t2);
\draw [->, line width=0.02cm] (s3) to [out=160, in=20] (v2);
 \draw [->, line width=0.02cm] (v2) to (v1);
 \draw [->, line width=0.05cm, color=blue] (u) to (v);
 \draw [->, line width=0.05cm, color=blue] (v1) to [out=-20, in=-160] (v3);

\end{tikzpicture}
\end{center}
\caption{A nice decomposition of a strong semicomplete digraph with four backwards arcs (in thin black). The arc $uv$ is regular-compulsory. The arc $v_1v_3$ is left-compulsory.}\label{fig:compulsory}
\end{figure}
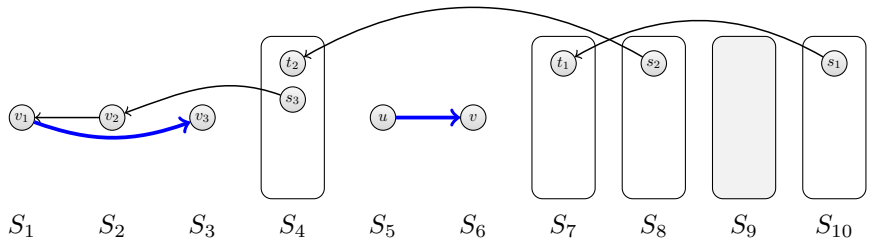

 \medskip

  \begin{theorem}\label{charac-unavoid}
 Let $D$ be a strong semicomplete digraph  of order at least $4$  and let $(S_1, \dots , S_p)$ be a nice decomposition of $D$.
 An arc is unavoidable if and only if it is either a cut-arc, regular-compulsory, left-compulsory, right-compulsory, or exceptional.
\end{theorem}

\begin{proof}
 
 If an arc $ad$ is exceptional, then Theorem~\ref{thm:charac-unavoid} implies that it is unavoidable ($R_1=\{c\}$ and $R_2=\{b\}$).
If $v_1 v_3$ is left-compulsory, then, again by Theorem~\ref{thm:charac-unavoid}, $v_1v_3$ is unavoidable ($R_1 = \{v_2\}$ and $R_2 = V(D) \setminus \{v_1,v_2,v_3\}$).
Analogously, if $v_{p-2}v_p$ is right-compulsory then it is unavoidable.
If $u v$ is regular-compulsory, then again by Theorem~\ref{thm:charac-unavoid}, $uv$ is unavoidable  ($R_1 = S_1 \cup \cdots \cup S_{i-1}$ and $R_2 = S_{i+2} \cup S_p$). 
 
 \medskip
 
 Let us now prove the reciprocal: if $uv$ is not a cut-arc (backward arc) and not exceptional, left-compulsory, right-compulsory or regular-compulsory, then
 it is not unavoidable.

  Note that Theorem~\ref{thm:charac-unavoid} implies that if $D$ is a strong semicomplete digraph and $uv \in A(D)$ is not a cut-arc then the 
following holds, where $D'=D\setminus{}uv$.  The arc $uv$ is unavoidable if and only if $N_{D'}^+(u) = N_{D'}^+(v)$ and
$N_{D'}^-(u) = N_{D'}^-(v)$ and $N_{D'}^+(u)  \cap N_{D'}^-(u) = \emptyset$ and there is only one arc from
$N_{D'}^+(u)$ to $N_{D'}^-(u)$.  
In particular if there is a path of length two between $u$ and $v$ then $uv$ is not unavoidable (unless it is a cut-arc). We shall use this observation several times below.

First assume that $uv$ is an unavoidable forward arc where $u \in S_i$ and $v \in S_j$. If $|S_i|>1$, then let $w \in S_i$  be
an out-neighbour of $u$. If $wv \in A(D)$, then $uwv$ is a path of length $2$ so $uv$ is not unavoidable, a contradiction.  
So $vw \in A(D)$ and $vw$ is a backward arc. 
In $D$ there must be two arc-disjoint paths, say $P_1$ and $P_2$,
 from $w$ to $u$ as otherwise there would be a cut-arc separating $w$ from $u$ which,
as $S_i$ is strong, must belong to $S_i$, a contradiction. Therefore there must be at least two arcs from $N^+(u)$ (as $w \in N^+(u)$)
to $N^-(u)$, so $uv$ is not unavoidable, a contradiction.  So $|S_i|=1$ and analogously $|S_j|=1$. 

 Assume that there is a backward arc $ru$ into $u$, where $r \in S_k$. As $vru$ is not a path we note that $rv \in A(D)$ and
therefore $i<r<j$ (as $S_k$ cannot have two backward arcs out of it). Assume that $i>1$ and let $xy$ be a backward arc
from $S_i \cup \cdots \cup S_p$ to $S_1 \cup \cdots \cup S_{i-1}$. As backward arcs are not nested (Proposition \ref{prop:nosame}) we see that $x$ must belong to $S_i \cup \cdots \cup S_{k-1}$. If $x = u$ then $uyv$ is a path, a contradiction, so
$x \in S_{i+1} \cup \cdots \cup S_{k-1}$. This implies that $xv \in A(D)$ (as otherwise $ru$ wouldn't be a cut-arc) and $uxv$ is a path, 
a contradiction. Therefore $i=1$. Analogously if there is a backward arc out of $v$ then $j=p$.

Now assume that there is a backward arc $vx$ out of $v$ and a backward arc, $yu$ into $u$. Then $i=1$ and $j=p$.
Note that $x \not=y$ as otherwise $vxu$ is a path. If there is any vertex in $w \in S_2 \cup \cdots \cup S_{p-1} \setminus \{x,y\}$
then $uwv$ is a path, so $V(D) = \{u,v,x,y\}$ and it is easy to see that $uv$ is an exceptional arc.  

So now assume that $u$ has a backward arc, $yu$, into it and $v$ has no backward arc out of it. Then $i=1$ and 
$S_1=\{u\}$, $S_2=\{y\}$ and $S_3=\{v\}$ as otherwise we could find a path of length two from $u$ to $v$. It is now easy to see that
$uv$ is left-compulsory.  Analogously if there is a backward arc out of $v$ but no backward arc into $u$, then $uv$ is
right-compulsory.

Finally assume that there is no backward arc into $u$ and no backward arc out of $v$. In this case $j=i+1$ as otherwise it 
is easy to find a path of length two from $u$ to $v$. We now see that $uv$ must be regular-compulsory.

\2

The remaining case is that $uv$ is a flat arc and $u,v \in S_i$. 
Then there are two arc-disjoint paths from $v$ to $u$ in $D-uv$ as otherwise there 
would be a cut-arc in $S_i$. But this implies that there are at least two arcs from $N^+(v)$ to $N^-(v)$, 
implying that $uv$ is not unavoidable.
  \end{proof}

\section*{Acknowledgements}
%%%%%%%%%%%%%%
This research was supported by the Danish research council under grant number DFF 7014-00037B and DISCO project, PICS, CNRS.

%\bibliography{refs}

\begin{thebibliography}{10}

\bibitem{bangJGT79}
J.~Bang-Jensen and A.Maddaloni.
\newblock Sufficient conditions for a digraph to be supereulerian.
\newblock {\em J. Graph Theory}, 79(1):8--20, 2015.

\bibitem{bang2009}
J.~Bang-Jensen and G.~Gutin.
\newblock {\em {Digraphs: Theory, Algorithms and Applications}}.
\newblock Springer-Verlag, London, 2nd edition, 2009.

\bibitem{bangCPC6}
J.~Bang-Jensen, G.~Gutin, and A.~Yeo.
\newblock {Hamiltonian cycles avoiding prescribed arcs in tournaments}.
\newblock {\em Combin. Prob. Comput.}, 6(3):255--261, 1997.

\bibitem{bang2018}
J.~Bang-Jensen and G.~(eds.) Gutin.
\newblock {\em {Classes of Directed Graphs}}.
\newblock Springer Monographs in Mathematics. Springer Verlag, London, 2018.

\bibitem{bangDM310}
J.~Bang-Jensen and T.~Jord\'an.
\newblock {Spanning 2-strong tournaments in 3-strong semicomplete digraphs}.
\newblock {\em Discrete Math.}, 310:1424--1428, 2010.

\bibitem{bangJA13}
J.~Bang-Jensen, Y.~Manoussakis, and C.~Thomassen.
\newblock {A polynomial algorithm for Hamiltonian-connectedness in semicomplete
  digraphs.}
\newblock {\em J. Algor.}, 13(1):114--127, 1992.

\bibitem{Cami59}
Paul Camion.
\newblock Chemins et circuits hamiltoniens des graphes complets.
\newblock {\em C. R. Acad. Sci. Paris}, 249:2151--2152, 1959.

\bibitem{fraisseGC3}
P.~Fraisse and C.~Thomassen.
\newblock {Hamiltonian dicycles avoiding prescribed arcs in tournaments}.
\newblock {\em Graphs Combin.}, 3(3):239--250, 1987.

\bibitem{guoDAM79b}
Y.~Guo.
\newblock {Spanning local tournaments in locally semicomplete digraphs}.
\newblock {\em Discrete Appl. Math.}, 79(1-3):119--125, 1997.

\bibitem{hoffman1960}
A.J. Hoffman.
\newblock {Some recent applications of the theory of linear inequalities to
  extremal combinatorial analysis}.
\newblock In R.~Bellman and M.~Hall, editors, {\em {Combinatorial Analysis}},
  pages 113--128. American Mathematical Society, Providence, RI, 1960.

\bibitem{Rede34}
L.~R\'{e}dei.
\newblock Ein kombinatorischer {S}atz.
\newblock {\em Acta. Litt. Sci. Szeged}, 7:39--43, 1934.

\bibitem{thomassenJCT28}
C.~Thomassen.
\newblock {Hamiltonian-connected tournaments}.
\newblock {\em J. Combin. Theory Ser. B}, 28(2):142--163, 1980.

\end{thebibliography}

\end{document}

\jbj{
  \begin{theorem}\cite{bangJGT79}\label{thm:SMDsupereuler}
    A strong semicomplete multipartite digraph has a spanning eulerian subdigraph if and only if it is strong and has an eulerian factor.
  \end{theorem}

  The following direct consequence of Theorem \ref{thm:SMDsupereuler} implies that Conjecture \ref{Euleravoidkarcs} holds when the $k$ arcs we wish to avoid induce a semicomplete subdigraph of $D$.

  \begin{corollary}
    \label{cor:del-clique}
    Let $D$ be a  semicomplete digraph and let $F$ be a set of arcs of $D$ that  a semicomplete subdigraph of $D$. If $\lambda{}(D)>|F|$, then $D\setminus F$ has a spanning eulerian subdigraph.
    \end{corollary}